\documentclass[journal]{IEEEtran}

\usepackage[T1]{fontenc}

\usepackage{amsmath}
\usepackage[cmintegrals]{newtxmath}
\usepackage{bm}
\usepackage{mathtools}
\usepackage{nccmath}

\newtheorem{theorem}{Theorem}
\newtheorem{lemma}{Lemma}



\usepackage{graphicx}
\usepackage{stfloats}
\usepackage{float}

\usepackage{tikz}
\usetikzlibrary{arrows.meta,calc}

\usepackage{algorithm}
\usepackage{algpseudocode}

\algrenewcommand\algorithmicrequire{\textbf{Input:}}
\algrenewcommand\algorithmicensure{\textbf{Output:}}

\usepackage{tabularx}
\usepackage[inline]{enumitem}

\usepackage[sort,compress]{cite}

\usepackage{xspace}
\usepackage{xcolor}
\usepackage{setspace}
\usepackage{soul}
\usepackage[normalem]{ulem}

\usepackage[user]{zref}
\usepackage[framemethod=tikz]{mdframed}

\usepackage{subcaption}

\usepackage{balance}

\usepackage[hidelinks]{hyperref}

\def\BibTeX{{\rm B\kern-.05em{\sc i\kern-.025em b}\kern-.08em
		T\kern-.1667em\lower.7ex\hbox{E}\kern-.125emX}}

\newcommand{\sot}[1]{}

\newcounter{revc}

\makeatletter

\zref@newprop{revcontent}{}
\zref@addprop{main}{revcontent}

\zref@newprop{revsec}{}
\zref@addprop{main}{revsec}

\zref@newprop{revpage}{}
\zref@addprop{main}{revpage}

\newcommand{\revi}[2]{%
	\zref@setcurrent{revsec}{\thesection}%
	\zref@setcurrent{revpage}{\thepage}%
	\zref@setcurrent{revcontent}{#2}%
	\refstepcounter{revc}%
	\label{#1}%
	\zlabel{#1}%
	\textcolor{blue}{#2}%
}

\newcommand{\revinu}[2]{%
	\zref@setcurrent{revsec}{\thesection}%
	\zref@setcurrent{revcontent}{#2}%
	\refstepcounter{revc}%
	\zlabel{#1}%
	\label{#1}%
	#2%
}

\newcommand{\revr}[2]{%
	\zref@setcurrent{revsec}{\thesection}%
	\zref@setcurrent{revcontent}{#2}%
	\refstepcounter{revc}%
	\zlabel{#1}%
	\label{#1}%
	\sot{#2}%
}

\makeatother

\expandafter\def\expandafter\quote\expandafter{%
	\quote
	\onehalfspacing
	\fontsize{12}{14}\selectfont
}

\definecolor{mycolor}{rgb}{0.122,0.435,0.698}

\newmdenv[
innerlinewidth=0.5pt,
roundcorner=4pt,
linecolor=mycolor,
innerleftmargin=6pt,
innerrightmargin=6pt,
innertopmargin=6pt,
innerbottommargin=6pt
]{mybox}

\begin{document}

\title{Energy Efficient Multi-User Beamforming and 3D Position Optimization for SIM-Assisted UAVs}  
\author{Chandan~Kumar~Sheemar,~\IEEEmembership{Member,~IEEE}, Giovanni Iacovelli,~\IEEEmembership{Member,~IEEE},\\ Sourabh Solanki,~\IEEEmembership{Member,~IEEE},  Wali Ullah Khan,~\IEEEmembership{Member,~IEEE}, \\George C. Alexandropoulos,~\IEEEmembership{Senior~Member,~IEEE}, and Symeon Chatzinotas,~\IEEEmembership{Fellow,~IEEE}   
 \thanks{C. K. Sheemar, G. Iacovelli, W. U. Khan, and S. Chatzinotas are with the SnT, University of Luxembourg, (emails:\{chandankumar.sheemar,\hspace{0pt}giovanni.iacovelli,\hspace{0pt}waliullah.khan,\hspace{0pt}symeon.chatzinotas\}@uni.lu). S. Solanki is with National Institute of Technology Warangal, 506004, India (e-mail: ssolanki@nitw.ac.in). G. C. Alexandropoulos is with the Department of Informatics and Telecommunications, National and Kapodistrian University of Athens, 16122 Athens, Greece (email: alexandg@di.uoa.gr).}
} 
 \maketitle
 
\begin{abstract}
This paper studies energy-efficient downlink multi-user transmissions with unmanned aerial vehicle (UAV) communication systems equipped with stacked intelligent metasurfaces (SIM), enabling wave-domain analog beamforming through multiple cascaded metasurface layers, while low-dimensional digital precoding is carried out using a limited number of transmit radio-frequency chains. This architecture enables flexible electromagnetic wave manipulation with reduced hardware complexity, making it particularly suitable for energy-constrained aerial platforms. We formulate a hardware-aware energy-efficiency (EE) maximization problem aiming to jointly optimize the digital precoder, the phase shifts of all SIM layers, and the three-dimensional UAV position under transmit-power, SIM operation, and UAV deployment constraints. The resulting problem is highly non-convex due to the fractional objective, the cascaded SIM structure and the unit-modulus phase constraints of the constituent metasurface layers, as well as the non-linear UAV-dependent channel. To address these challenges, we develop a transform-based alternating optimization framework that combines Dinkelbach’s method, dual and quadratic transforms, o enable closed-form digital beamforming, Riemannian manifold optimization for SIM phase shifts, and successive convex approximation (SCA) for UAV positioning. Convergence and complexity analyses are provided to characterize the proposed algorithm. The presented numerical results showcase that the proposed joint design significantly improves EE compared with fully digital and maximum ratio transmission benchmark schemes, while revealing important design trade-offs among transmit power, SIM size, and the number of its constituent stacked layers.
\end{abstract}

\IEEEpeerreviewmaketitle

\section{Introduction} \label{Intro}

\IEEEPARstart{U}{nmanned} aerial vehicles (UAVs) are a promising enabler for next-generation wireless networks due to their flexible deployment, strong line-of-sight (LoS) connectivity, and rapid on-demand coverage capabilities. These features make UAVs particularly attractive for scenarios such as temporary hotspot coverage, emergency communications, and dynamic network densification. However, supporting multi-user communications with large antenna arrays onboard UAVs remains a major challenge. In particular, fully digital beamforming requires a dedicated radio frequency (RF) chain per antenna element, leading to high hardware complexity, power consumption, and payload requirements \cite{liu2023deployment,sheemar2025joint,sheemar2022practical,sheemar2024parallel}. Such requirements are difficult to accommodate in UAV platforms, where size, weight, and energy constraints are critical. This motivates the development of lightweight and energy-efficient transmission architectures tailored to aerial systems \cite{sheemar2026joint}.

In this context, reconfigurable metasurfaces have recently emerged as a transformative technology for wireless communications \cite{khan2024reconfigurable,li2025ris}. By employing meta elements with tunable electromagnetic responses, metasurfaces enable dynamic control over the amplitude, phase, and polarization of incident waves, allowing programmable manipulation of wireless propagation environments \cite{liu2021reconfigurable,elmossallamy2020reconfigurable}. Early works have focused on single-layer reflecting intelligent  surfaces (RIS), which can reconfigure signal reflections to enhance coverage and link quality \cite{huang2019reconfigurable,sheemar2023full,khan2024beyond,sheemar2023irs}. However, single-layer architectures are inherently limited in their signal processing capabilities, as they can only apply element-wise phase shifts without enabling more advanced wave transformations.

To overcome the inherent limitations of single-layer metasurfaces, stacked intelligent metasurfaces (SIMs) have recently been introduced as a powerful and flexible extension of conventional metasurface architectures \cite{sheemar2026survey}. Unlike traditional metasurfaces that apply a single stage of phase manipulation to the impinging electromagnetic wave, a SIM comprises multiple closely spaced and electromagnetically coupled metasurface layers. Each layer consists of a large number of nearly passive meta elements, whose responses can be individually programmed to impose controllable phase shifts on the incident signal. The proximity between layers allows the electromagnetic wave to undergo successive transformations as it propagates through the stack, resulting in a cascaded wave processing effect. This multi-layer structure enables SIMs to perform multi-stage analog signal processing directly in the electromagnetic domain. In particular, the wave interactions across the layers can be interpreted as a sequence of linear transformations, where each metasurface layer acts as a diagonal phase-shifting matrix, while the free-space propagation between layers introduces structured coupling matrices. By appropriately configuring these layers, the overall SIM architecture can approximate a wide class of linear precoding and combining operations, effectively emulating high-dimensional digital beamforming functionalities. This perspective establishes SIMs as a form of wave-based analog computing platform, capable of implementing complex signal processing tasks with minimal hardware overhead.

As a consequence, SIMs can support advanced transmission functionalities such as directional analog beamforming, spatial multiplexing for multi-user communications, and interference suppression through constructive and destructive wave superposition. Unlike conventional hybrid or fully digital beamforming architectures, these operations are realized intrinsically through the physics of wave propagation and scattering, rather than through power-hungry digital processing chains. Importantly, since the majority of signal manipulation is performed in the passive metasurface domain, the number of required RF chains can be drastically reduced, often to a value comparable only to the number of data streams or users. From a system-level perspective, this reduction in RF hardware translates directly into lower circuit power consumption, reduced heat dissipation, and lighter system weight, all of which are critical advantages for practical deployments. Furthermore, the reduction of complex digital signal processing units and high-resolution data converters simplifies the overall system design and improves energy efficiency (EE). Therefore, SIM-based architectures offer a compelling alternative to conventional fully digital or hybrid systems, particularly in scenarios where power, cost, and form-factor constraints are stringent. These benefits make SIMs especially attractive for emerging applications such as UAV-assisted communications, where lightweight, EE, hardware costs, and high-performance transmission solutions are essential \cite{an2023stacked,an2024stacked_WC,papazafeiropoulos2025ergodic_outage_3}.

\subsection{State-of-the-Art and Motivation}
Recent advances in SIMs have rapidly expanded across multiple dimensions of wireless system design. A fundamental challenge in SIM-assisted systems lies in accurate channel estimation under high-dimensional wave-domain transformations and limited RF chains. Model-driven approaches such as subspace estimation \cite{yao2024channel} and LoS-aware MMSE designs under Rician fading \cite{papazafeiropoulos2025channel} have demonstrated that embedding propagation structure into the estimation process significantly improves accuracy. To further exploit the intrinsic multilayer coupling of SIMs, tensor-based representations \cite{ginige2025nested} and hybrid wave-domain/digital estimation frameworks \cite{an2024hybrid} have been proposed, enabling joint estimation of channel and metasurface responses. In parallel, sparsity-driven techniques tailored to mmWave and near-field regimes, including angular-domain compressed sensing \cite{yao2024sparse} and polar-domain Bayesian learning \cite{yao2025sparse}, effectively leverage structured propagation characteristics to reduce training overhead. Complementing these model-based approaches, learning-driven methods such as FlatCE-Net \cite{dong2026deep} and convolutional architectures \cite{lawal2025channel} have shown strong robustness in high-dimensional scenarios where statistical channel knowledge is unavailable. Collectively, these works highlight that accurate CSI acquisition in SIM systems can be attained to enable wave-domain signal processing.

On the performance optimization side, a growing body of literature has established both analytical foundations and practical designs for SIM-assisted communications. Early works on achievable rate optimization under statistical channel state information (CSI) \cite{papazafeiropoulos2024achievable} and fully programmable wave-domain transceivers \cite{papazafeiropoulos2024achievable_2} demonstrated that SIMs can emulate high-dimensional precoding with significantly reduced RF complexity. Analytical studies, including ergodic mutual information and outage characterizations \cite{papazafeiropoulos2025ergodic_outage_3} as well as capacity bounds for holographic multiple-input-multiple-output (MIMO) \cite{papazafeiropoulos2025ergodic_HMIMO}, have provided key insights into the impact of stacking depth and channel conditions on system performance. Practical optimization frameworks based on alternating optimization and successive convex approximation (SCA) \cite{bahingayi2025refined} have enabled efficient joint design of SIM configurations and digital beamforming. Moreover, hybrid architectures combining passive and active metasurface layers \cite{darsena2025design,darsena2024downlink} have been shown to mitigate propagation losses and further enhance achievable rates. Importantly, recent scaling analyses \cite{bahingayi2025scaling,papazafeiropoulos2025performance} indicate that increasing the number of SIM layers consistently improves interference suppression and throughput, challenging the conventional notion of diminishing returns in metasurface design.

Emerging works have also addressed practical constraints and system-level considerations. Statistical CSI-based designs \cite{xia2025statistical} reduce the burden of instantaneous channel acquisition, while impairment-aware optimization \cite{rezvani2025uplink} accounts for hardware non-idealities inherent to metasurface implementations. Fairness-oriented frameworks, including max--min optimization \cite{ginige2025max} and joint fairness--throughput tradeoff designs \cite{fang2025stacked}, demonstrate that SIMs provide additional degrees of freedom to balance user performance without significant spectral efficiency loss.

Despite the rapid progress in SIM-enabled wireless systems, the integration of SIMs with UAV communications has received relatively limited attention, with only a few studies such as \cite{fan2025joint,zarini2025orchestration} addressing this emerging topic. In \cite{fan2025joint}, the authors investigate an uplink SIM-assisted UAV system and propose a joint optimization framework for user association, UAV placement, and SIM phase shifts combining. Similarly, \cite{zarini2025orchestration} considers the joint design of SIM configuration and UAV control based on a learning-driven approach. However, both works primarily focus on sum-rate maximization, do not consider a scalable hybrid transceiver design and overlook the EE considerations, which are particularly critical in SIM-assisted UAV systems due to the limited onboard battery capacity, stringent payload and hardware constraints, and the additional power consumption associated with UAV propulsion and RF processing. Moreover, the adoption of hybrid beamforming architectures with SIMs further motivates energy-efficient design, since part of the beamforming functionality can be shifted from power-hungry digital RF chains directly to low-power electromagnetic wave-domain processing.

 \subsection{Main Contributions}

To address this research gap, we develop a unified optimization framework for the downlink of a multi-user UAV system equipped with SIM. In the considered architecture, SIM enables analog wave-domain beamforming for multi-user communications through cascaded metasurface layers, while a low-dimensional digital precoder is employed at the RF-chain side. Unlike conventional UAV beamforming designs that mainly focus on spectral efficiency, we formulate a hardware-aware EE maximization problem that jointly accounts for radiated transmit power, RF-chain power, baseband processing power, element-wise SIM control power, and additional control circuitry. The resulting formulation jointly optimizes the digital beamformer, the phase shifts of all SIM layers, and the three-dimensional UAV position under transmit-power, unit-modulus, and UAV deployment constraints.

To solve the resulting highly non-convex optimization problem, we propose a transform-based alternating optimization framework. First, Dinkelbach's method is applied to convert the fractional EE objective into a subtractive form. Then, dual and quadratic transforms are used to decouple the logarithmic signal-to-interference-plus-noise (SINR) terms, enabling efficient block-wise updates. For fixed SIM phases and UAV position, the digital beamformer is obtained in closed form from the Karush-Kuhn-Tucker (KKT) conditions, with a bisection search used to satisfy the transmit-power constraint. The SIM phase shifts are optimized through Riemannian manifold ascent over the product of complex unit circles, while the UAV position is updated using a successive convex approximation (SCA) method based on analytically derived channel gradients. We also provide convergence and complexity analyses, showing that the proposed algorithm monotonically improves the transformed objective and converges to a stationary solution. Simulation results confirm that the proposed joint design achieves significant energy-efficiency gains over fully digital and SIM-assisted maximum ratio transmission (MRT) benchmarks, while also revealing key design tradeoffs with respect to transmit power, SIM size, and number of stacked layers.

 \emph{Notations:} Scalars, vectors, and matrices are denoted by italic, bold lowercase, and bold uppercase letters, respectively. $(\cdot)^T$, $(\cdot)^H$, and $(\cdot)^*$ denote transpose, Hermitian transpose, and complex conjugate, respectively. $\|\cdot\|_2$ denotes the Euclidean norm, while $\|\cdot\|_F$ denotes the Frobenius norm. $\mathbb{E}[\cdot]$ denotes statistical expectation, and $\mathcal{CN}(0,\sigma^2)$ represents a circularly symmetric complex Gaussian distribution with zero mean and variance $\sigma^2$. The operator $\mathrm{diag}(\cdot)$ maps a vector to a diagonal matrix, and $\mathrm{Tr}(\cdot)$ denotes the trace of a matrix. The real part of a complex quantity is denoted by $\Re\{\cdot\}$. The notation $[\mathbf{x}]_m$ denotes the $m$-th element of vector $\mathbf{x}$, and $\mathbf{I}$ denotes the identity matrix of appropriate dimension.

\emph{Paper Organization:} The rest of the paper is organized as follows. Section \ref{sec2} introduces the system model and formulates the energy-efficiency maximization problem. Section \ref{sec3} presents the proposed optimization framework, including the transformation techniques and iterative algorithm. Section \ref{sec4} provides numerical results and discussions, and Section \ref{sec5} concludes the paper.

\section{System Model and Problem Formulation} \label{sec2}

 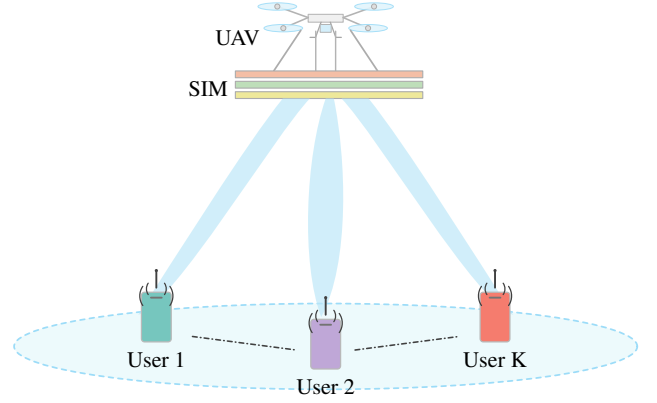
\begin{figure}[t]
\centering
\resizebox{0.95\linewidth}{!}{
\begin{tikzpicture}[line cap=round,line join=round,>=Latex]

\definecolor{beamblue}{RGB}{190,232,248}
\definecolor{usergreen}{RGB}{118,198,190}
\definecolor{userpurple}{RGB}{191,167,215}
\definecolor{userred}{RGB}{245,124,110}
\definecolor{simone}{RGB}{244,190,165}
\definecolor{simtwo}{RGB}{190,222,185}
\definecolor{simthree}{RGB}{238,232,154}
\definecolor{groundfill}{RGB}{232,249,252}

\fill[groundfill,opacity=0.75] (0,-4.1) ellipse (6.9 and 0.95);
\draw[cyan!35,dashed,line width=1.1pt] (0,-4.1) ellipse (6.9 and 0.95);

\node[font=\Large] at (-1.95,2.75) {UAV};

\begin{scope}[shift={(0,3.08)},scale=1.02]
    \draw[gray!55,fill=gray!12,line width=0.7pt] (-0.38,0.0) rectangle (0.38,0.17);
    \draw[gray!70,line width=0.8pt] (-0.1,0.0) -- (-0.22,-0.33);
    \draw[gray!70,line width=0.8pt] (0.1,0.0) -- (0.22,-0.33);
    \draw[gray!70,line width=0.8pt] (-0.22,-0.33) -- (-0.34,-0.33);
    \draw[gray!70,line width=0.8pt] (0.22,-0.33) -- (0.34,-0.33);
    \draw[gray!60,fill=cyan!15,line width=0.6pt] (-0.12,-0.05) rectangle (0.12,-0.22);

    \draw[gray!60,line width=1pt] (-0.38,0.1) -- (-1.05,0.34);
    \draw[gray!60,line width=1pt] (0.38,0.1) -- (1.05,0.34);
    \draw[gray!60,line width=1pt] (-0.38,0.1) -- (-0.92,-0.13);
    \draw[gray!60,line width=1pt] (0.38,0.1) -- (0.92,-0.13);

    \foreach \x/\y in {-1.05/0.34,1.05/0.34,-0.92/-0.13,0.92/-0.13}{
        \draw[cyan!35,fill=cyan!10,opacity=0.9] (\x,\y) ellipse (0.42 and 0.08);
        \draw[gray!65,fill=gray!25] (\x,\y) circle (0.055);
    }
\end{scope}

\draw[gray!65,line width=1.1pt] (-0.55,2.90) -- (-1.15,2.00);
\draw[gray!65,line width=1.1pt] (0.55,2.90) -- (1.15,2.00);
\draw[gray!60,line width=0.9pt] (-0.22,2.88) -- (-0.22,2.00);
\draw[gray!60,line width=0.9pt] (0.22,2.88) -- (0.22,2.00);

\node[font=\Large] at (-2.6,1.6) {SIM};

\draw[fill=simone,draw=gray!65,line width=0.7pt] (-2.0,1.85) rectangle (2.15,2.00);
\draw[fill=simtwo,draw=gray!65,line width=0.7pt] (-2.0,1.62) rectangle (2.15,1.77);
\draw[fill=simthree,draw=gray!65,line width=0.7pt] (-2.0,1.39) rectangle (2.15,1.54);

\fill[beamblue,opacity=0.70] (-0.95,1.39) .. controls (-1.75,0.4) and (-3.35,-2.15) .. (-4.15,-3.4)
                             .. controls (-3.35,-3.15) and (-1.15,0.55) .. (-0.35,1.39) -- cycle;

\fill[beamblue,opacity=0.70] (0.0,1.39) .. controls (-0.45,0.10) and (-0.55,-2.25) .. (-0.10,-3.45)
                            .. controls (0.55,-2.25) and (0.45,0.10) .. (0.18,1.39) -- cycle;

\fill[beamblue,opacity=0.70] (0.95,1.39) .. controls (1.75,0.4) and (3.35,-2.15) .. (4.15,-3.4)
                            .. controls (3.35,-3.15) and (1.15,0.55) .. (0.35,1.39) -- cycle;

\newcommand{\phone}[4]{
\begin{scope}[shift={(#1,#2)}]
    \draw[fill=#3,draw=gray!55,line width=0.8pt,rounded corners=2pt]
        (-0.33,-0.55) rectangle (0.33,0.55);
    \draw[black!55,line width=1.2pt] (-0.12,0.43) -- (0.12,0.43);
    \node[font=\Large] at (0,-0.95) {#4};
    \draw[black!75,line width=1.2pt] (0,0.65) -- (0,1.02);
    \fill[black!75] (0,1.03) circle (0.035);
    \draw[black!75,line width=0.8pt] (-0.18,0.75) arc (145:215:0.22);
    \draw[black!75,line width=0.8pt] (0.18,0.75) arc (35:-35:0.22);
    \draw[black!75,line width=0.8pt] (-0.30,0.68) arc (145:215:0.37);
    \draw[black!75,line width=0.8pt] (0.30,0.68) arc (35:-35:0.37);
\end{scope}
}

\phone{-3.75}{-3.45}{usergreen}{User 1}
\phone{0}{-4.05}{userpurple}{{User 2}}
\phone{3.75}{-3.45}{userred}{User K}

\draw[black!75,dash dot,line width=1pt] (-2.95,-3.88) -- (-0.65,-4.18);
\draw[black!75,dash dot,line width=1pt] (0.65,-4.18) -- (2.95,-3.88);

\end{tikzpicture}
}

\caption{SIM-assisted UAV communications.}
\label{fig1}
\end{figure}

We consider the downlink of a multi-user communication system, comprising a UAV serving $K$ single-antenna ground users, as illustrated in Fig.~\ref{fig1}. The UAV acts as an aerial transmitter and is equipped with SIM, which enables flexible wave-domain manipulation of electromagnetic signals. The transmitter adopts a hybrid beamforming architecture that combines digital precoding in the baseband domain with analog beamforming implemented via SIM.  

Let $N$ denote the number of transmit radio-frequency (RF) chains at the UAV's transmitter. To simplify the system design and avoid additional digital multiplexing overhead, we assume $N=K$, i.e., the number of RF chains is equal to the number of users. The SIM consists of $L$ metasurface layers indexed by $\mathcal{L}=\{1,\ldots,L\}$, where each layer is composed of $M$ passive meta-atoms arranged in a uniform planar array (UPA). We assume $M\geq N$ to guarantee sufficient spatial degrees of freedom for beam shaping and interference suppression. Let $\mathcal{M}=\{1,\ldots,M\}$ denote the index set of meta-atoms. Specifically, the transmission coefficient of the $m$-th meta-atom on the $\ell$-th metasurface layer is modelled as
$
\phi_{\ell,m}=e^{j\theta_{\ell,m}}$, where $\theta_{\ell,m}\in[0,2\pi),
$ is the controllable phase shift.  For compact representation, we define
$
\boldsymbol{\phi}_\ell=
[\phi_{\ell,1},\phi_{\ell,2},\ldots,\phi_{\ell,M}]^T
$
and
$
\mathbf{\Phi}_\ell=\mathrm{diag}(\boldsymbol{\phi}_\ell)\in\mathbb{C}^{M\times M},
$
which represent the transmission coefficient vector and diagonal phase-shift matrix of the $\ell$-th SIM layer, respectively.

All metasurface layers are assumed to follow an identical lattice structure and are modeled as UPAs \cite{sheemar2026survey}, as shown in Fig.~\ref{fig:sim_geometry}. The spatial configuration of the meta-atoms determines the electromagnetic coupling across the SIM. In particular, the distance between the $m$-th and $\tilde m$-th meta-atoms on the same transmit metasurface is given by
$
r_{m,\tilde m}= \Delta \sqrt{(m_z-\tilde m_z)^2+(m_x-\tilde m_x)^2},
$
where $\Delta$ denotes the inter-element spacing. The indices of the $m$-th meta-atom are expressed as
$
m_z=\left\lceil\frac{m}{m_{\max}}\right\rceil,
$
and
$
m_x=\mathrm{mod}(m-1,m_{\max})+1,
$
where $m_{\max}$ is the number of elements along one dimension of the UPA, such that $M=m_{\max}^2$. This indexing facilitates mapping between the one-dimensional element index and the two-dimensional grid structure. The metasurface layers are assumed to be parallel and equally spaced. Let $D_t$ denote the total thickness of the SIM and $d_t=D_t/(L-1)$ the spacing between adjacent layers. The propagation distance between the $\tilde m$-th meta-atom on layer $\ell-1$ and the $m$-th meta-atom on layer $\ell$ is therefore
$
r_{m,\tilde m}^{(\ell)}=\sqrt{r_{m,\tilde m}^2+d_t^2},
$
with $\ell\in\mathcal{L}\setminus\{1\}$. 

\begin{figure}[t]
\centering

\begin{subfigure}[t]{0.32\textwidth}
\centering
\resizebox{\linewidth}{!}{%
\begin{tikzpicture}[>=Latex, font=\Large]

\definecolor{darkgrayline}{RGB}{45,45,45}
\definecolor{panelgray}{RGB}{232,232,232}
\definecolor{atomgray}{RGB}{185,190,198}
\definecolor{softgray}{RGB}{245,245,245}

\tikzset{
  atom/.style={draw=darkgrayline, line width=1pt, fill=atomgray,
               minimum width=9mm, minimum height=9mm, inner sep=0pt},
}

\node[anchor=west] at (0,6.25) {The $\ell$-th SIM layer};

\draw[fill=panelgray, draw=darkgrayline, line width=0.8pt] (0.2,0.6) rectangle (6.2,5.6);

\node[atom] at (1.0,4.9) {$1$};
\node[atom] at (2.2,4.9) {$2$};
\node at (3.6,4.9) {$\cdots$};
\node[atom] (topright) at (5.4,4.9) {$m_{\max}$};

\node[atom] at (1.0,3.6) {};
\node[atom] at (2.2,3.6) {};
\node at (3.6,3.6) {$\cdots$};
\node[atom] at (5.4,3.6) {};

\node at (1.0,2.6) {$\vdots$};
\node at (2.2,2.6) {$\vdots$};
\node at (5.4,2.6) {$\vdots$};

\node[atom] at (1.0,1.4) {};
\node[atom] at (2.2,1.4) {};
\node at (3.6,1.4) {$\cdots$};
\node[atom] at (5.4,1.4) {$M$};

\draw[densely dashed, ->, line width=1.0pt, darkgrayline] (topright.east) -- ++(1.1,0)
  node[right] {$m_{\max}$};

\draw[densely dashed, darkgrayline] (1.0,0.55) -- (1.0,0.05);
\draw[densely dashed, darkgrayline] (2.2,0.55) -- (2.2,0.05);
\draw[<->, line width=1.0pt, darkgrayline] (1.0,0.25) -- (2.2,0.25)
  node[midway, below] {$\Delta$};

\end{tikzpicture}%
}
\caption{The $\ell$-th SIM layer comprising $M=m_{\max}^2$ meta-atoms.}
\label{fig_SIM_conf_1}
\end{subfigure}
\hfill
\begin{subfigure}[t]{0.44\textwidth}
\centering
\resizebox{\linewidth}{!}{%
\begin{tikzpicture}[>=Latex, font=\Large]

\definecolor{darkgrayline}{RGB}{45,45,45}
\definecolor{layergray}{RGB}{220,220,220}
\definecolor{atomgray}{RGB}{185,190,198}
\definecolor{feedgray}{RGB}{80,80,80}

\tikzset{
  layerline/.style={draw=darkgrayline, line width=1.25pt},
  port/.style={draw=darkgrayline, line width=1.3pt, fill=white},
  atom/.style={draw=darkgrayline, line width=1.2pt, fill=atomgray,
               minimum width=9mm, minimum height=9mm, inner sep=0pt},
  dashedbox/.style={draw=darkgrayline, line width=1.1pt, dash pattern=on 7pt off 5pt},
}

\node[font=\Large] at (7.2,6.75) {TX-SIM};

\draw[layerline] (0.7,5.6) -- (2.0,6.1);
\draw[layerline] (0.7,2.2) -- (2.0,1.7);
\draw[layerline] (2.0,6.1) -- (3.0,6.1);
\draw[layerline] (2.0,1.7) -- (3.0,1.7);

\foreach \y/\lab in {5.2/1,4.3/2,2.4/N}{
  \draw[port] (1.0,\y) circle (0.16);
  \node[anchor=east] at (0.7,\y) {$\lab$};
}
\node at (0.85,3.35) {$\vdots$};
\node[anchor=west] at (0.25,6.25) {RF feeds};

\draw[layerline, fill=layergray!25] (3.0,1.7) rectangle (12.8,6.1);

\foreach \x in {4.6,6.2,7.8,9.4,11.0}{
  \draw[layerline] (\x,1.7) -- (\x,6.1);
}

\node[below] at (3.0,1.7) {$1$};
\node[below] at (6.2,1.7) {$\ell-1$};
\node[below] at (7.8,1.7) {$\ell$};
\node[below] at (12.8,1.7) {$L$};

\draw[feedgray,line width=1pt,->] (2.2,5.2) -- (3.0,5.2);
\draw[feedgray,line width=1pt,->] (2.2,4.3) -- (3.0,4.3);
\draw[feedgray,line width=1pt,->] (2.2,2.4) -- (3.0,2.4);

\draw[dashedbox] (5.4,2.3) rectangle (8.6,5.5);

\node[atom] (mtil) at (6.2,4.9) {$\tilde m$};
\node[atom] (mL)   at (6.2,2.9) {$m$};
\node[atom] (mR)   at (7.8,2.9) {$m$};

\draw[densely dashed, darkgrayline] (5.75,4.8) -- (5.75,3.0);
\draw[<->, line width=1.0pt, darkgrayline] (5.75,4.7) -- (5.75,3.1);
\node[anchor=east] at (5.65,3.9) {$r_{m,\tilde m}$};

\draw[densely dashed, darkgrayline] (mtil.south east) -- (mR.north west);
\node[anchor=west] at (6.95,3.95) {$r_{m,\tilde m}^{(\ell)}$};

\draw[<->, line width=1.0pt, darkgrayline] ($(mL.south)+(0,-0.35)$) -- ($(mR.south)+(0,-0.35)$)
  node[midway,below] {$d_t$};

\node at (10.2,3.9) {$\cdots$};

\end{tikzpicture}%
}
\caption{The overall transmit SIM structure and its inter-layer propagation geometry.}
\label{fig_SIM_conf}
\end{subfigure}

\caption{The considered TX-SIM at the UAV.}
\label{fig:sim_geometry}
\end{figure}
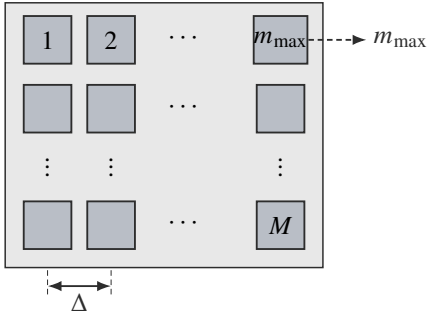
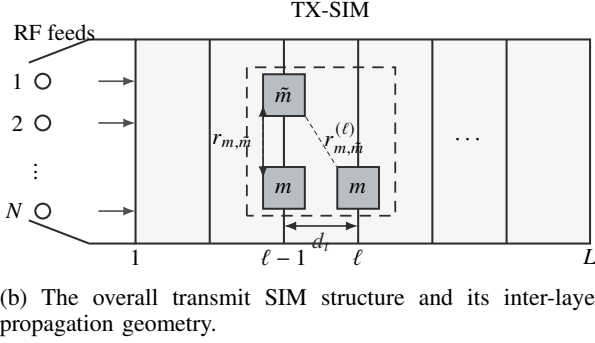

To accurately model the electromagnetic wave propagation across layers, we adopt the Rayleigh--Sommerfeld diffraction formulation. Accordingly, the propagation coefficient between these two meta-atoms is given by \cite{an2023stacked}
\begin{equation}
w_{m,\tilde m}^{\ell}=
\frac{A_t\cos\chi_{m,\tilde m}^{\ell}}{r_{m,\tilde m}^{(\ell)}}
\left(
\frac{1}{2\pi r_{m,\tilde m}^{(\ell)}}
-\frac{j}{\lambda}
\right)
e^{j\frac{2\pi}{\lambda}r_{m,\tilde m}^{(\ell)}} ,
\end{equation}
where $A_t$ denotes the physical area of each meta-atom, $\lambda$ is the carrier wavelength, and $\chi_{m,\tilde m}^{\ell}$ is the angle between the propagation direction and the surface normal. This expression captures both amplitude decay and phase accumulation as waves propagate through the layered structure. Let $\mathbf{W}_1\in\mathbb{C}^{M\times N}$ denote the propagation matrix from the $N$ RF feeds to the first metasurface layer, which accounts for feed-to-surface coupling. For $\ell\ge2$, let $\mathbf{W}_\ell\in\mathbb{C}^{M\times M}$ denote the transmission matrix between consecutive metasurface layers, whose entries are given by $w_{m,\tilde m}^{\ell}$. By cascading the effects of all layers, the overall SIM analog precoder can be expressed as
$
\boldsymbol{\Psi}
=
\mathbf{\Phi}_L\mathbf{W}_L
\cdots
\mathbf{\Phi}_2\mathbf{W}_2
\mathbf{\Phi}_1\mathbf{W}_1
\in\mathbb{C}^{M\times N}.
$
This cascaded structure highlights the wave-domain processing capability of the SIM, where each layer progressively reshapes the electromagnetic field. Let $\mathbf{s}\in\mathbb{C}^{K\times1}$ denote the transmitted symbol vector satisfying $\mathbb{E}[\mathbf{s}\mathbf{s}^H]=\mathbf{I}$. Digital beamforming is performed using the precoding matrix
$
\mathbf{V}\in\mathbb{C}^{N\times K}.
$
After digital precoding and SIM processing, the transmitted signal becomes
$
\mathbf{x} =\boldsymbol{\Psi}\mathbf{V}\mathbf{s}.
$

We assume perfect channel state information (CSI). Let $\mathbf{h}_k(\mathbf{q})\in\mathbb{C}^{M\times1}$ denote the channel between the SIM output aperture and user $k$, where $\mathbf{q} =[x_u,y_u,z_u]^T$ represents the UAV position and $\mathbf{u}_k =[x_k,y_k,0]^T$ denotes the $k$-th user position. The UAV altitude introduces an additional degree of freedom that impacts both large-scale fading and angular characteristics of the channel. Under the far-field assumption, the channel is modeled as a Rician fading channel that captures both line-of-sight (LoS) and scattered components. Let $d_k(\mathbf{q}) = \| \mathbf{q}-\mathbf{u}_k \|_2$ denote the distance between the UAV and user $k$. Then, the channel coefficient associated with the $m$-th meta-atom is
\begin{equation} \label{canale}
h_{k,m}(\mathbf{q})=
\sqrt{\beta_k(\mathbf{q})}
\left(
\sqrt{\frac{K_k}{K_k+1}}\, a_m(\vartheta_k,\varphi_k)
+
\sqrt{\frac{1}{K_k+1}}\, z_{k,m}
\right),
\end{equation}
where $\beta_k(\mathbf{q})=\beta_0 d_k(\mathbf{q})^{-\alpha}$ models large-scale path loss with reference gain $\beta_0$ and path-loss exponent $\alpha$. The parameter $K_k$ denotes the Rician factor, which quantifies the relative strength of the LoS component. The term $a_m(\vartheta_k,\varphi_k)$ represents the $m$-th element of the SIM array response toward user $k$, determined by the azimuth and elevation angles $(\vartheta_k,\varphi_k)$, while $z_{k,m}\sim\mathcal{CN}(0,1)$ models the non-LoS scattering component. Accordingly, the channel vector is expressed as
$
\mathbf{h}_k(\mathbf{q})=
\left[
h_{k,1}(\mathbf{q}),\,
h_{k,2}(\mathbf{q}),\,
\ldots,\,
h_{k,M}(\mathbf{q})
\right]^T
\in\mathbb{C}^{M\times1}.
$ The received signal at user $k$ is therefore
\begin{equation} \label{received_signal}
r_k=\mathbf{h}_k^H(\mathbf{q})\boldsymbol{\Psi}\mathbf{V}\mathbf{s}+n_k ,
\end{equation}
where $n_k\sim\mathcal{CN}(0,\sigma_k^2)$ denotes additive white Gaussian noise (AWGN). The received signal consists of the desired signal component as well as multi-user interference caused by imperfect spatial separation. Let $\mathbf{v}_k$ denote the digital beamformer of the $k$-th user. The resulting SINR of user $k$ is given by
\begin{equation}
\gamma_k(\mathbf{V},\{\mathbf{\Phi}_\ell\},\mathbf{q})=
\frac{|\mathbf{h}_k^H(\mathbf{q})\boldsymbol{\Psi}\mathbf{v}_k|^2}
{\sum_{i\neq k}|\mathbf{h}_k^H(\mathbf{q})\boldsymbol{\Psi}\mathbf{v}_i|^2+\sigma_k^2},
\end{equation}
which explicitly depends on the digital precoder $\mathbf{V}$, the SIM phase shifts $\{\mathbf{\Phi}_\ell\}$, and the UAV position $\mathbf{q}$. This coupling highlights the joint design nature of the system, where electromagnetic wave manipulation, signal processing, and UAV deployment must be optimized in a unified framework.

 \subsection{Problem Formulation}

For the considered SIM-assisted hybrid beamforming architecture, the achievable sum rate of the system is defined as
\begin{equation}
R(\mathbf{V},\{\mathbf{\Phi}_\ell\},\mathbf{q})
=
\sum_{k=1}^{K}\log_2\!\Big(1+\gamma_k(\mathbf{V},\{\mathbf{\Phi}_\ell\},\mathbf{q})\Big),
\end{equation}
where $\gamma_k(\mathbf{V},\{\mathbf{\Phi}_\ell\},\mathbf{q})$ denotes the SINR of $k$-th user. While maximizing the sum rate is a common objective, it does not explicitly account for the power consumption associated with both signal transmission and hardware operation. To consider that effect, the total power consumption for signal processing is modeled as
$
P_{\rm tot}(\mathbf{V},\{\mathbf{\Phi}_\ell\})
=
\xi\,\mathrm{Tr}\!\left(\boldsymbol{\Psi}\mathbf{V}\mathbf{V}^H\boldsymbol{\Psi}^H\right)
+
P_{\rm c},
$
where the first term represents the effective transmit power after SIM processing, scaled by the inverse power-amplifier efficiency $\xi\ge 1$. The trace term $\mathrm{Tr}(\boldsymbol{\Psi}\mathbf{V}\mathbf{V}^H\boldsymbol{\Psi}^H)$ corresponds to the total radiated power at the SIM output aperture. The second term, $P_{\rm c}$, accounts for the circuit power consumption, which is independent of the instantaneous transmit signal but depends on the system architecture. To explicitly capture the hardware cost of the SIM-based hybrid beamforming structure, the circuit power is modeled as
$
P_{\rm c}= N P_{\rm RF}+P_{\rm BB}+LM P_{\rm E}+P_{\rm ctrl},
$
where $P_{\rm RF}$ denotes the power consumption of each RF chain, $P_{\rm BB}$ is the power required for baseband signal processing, and $P_{\rm E}$ represents the control power consumed by each meta-atom. The term  $P_{\rm E}$ denotes the per element power consumption of SIM and $LM P_{\rm E}$ counts for the total power consumed to tune all the SIM elements.  Finally, $P_{\rm ctrl}$ accounts for additional power consumption due to supporting circuitry such as controllers, biasing networks, and signal routing components. 

Based on the above definitions, the system EE is defined as the ratio between the achievable sum rate and the total power consumption, given by
\begin{equation}
\eta(\mathbf{V},\{\mathbf{\Phi}_\ell\},\mathbf{q})
=
\frac{B  R(\mathbf{V},\{\mathbf{\Phi}_\ell\},\mathbf{q})}
{P_{\rm tot}(\mathbf{V},\{\mathbf{\Phi}_\ell\})},
\end{equation}
where $B$ denotes the system bandwidth. This metric, measured in bits/Joule, quantifies how efficiently the system converts consumed energy into useful information transmission.
Accordingly, the joint design of digital beamforming, SIM phase shifts, and UAV placement can be formulated as the following EE maximization problem:
 \begin{subequations}\label{eq:ee_problem}
\begin{align}
\max_{\mathbf{V},\,\{\mathbf{\Phi}_\ell\},\,\mathbf{q}}
\quad &
\eta(\mathbf{V},\{\mathbf{\Phi}_\ell\},\mathbf{q})
\label{eq:ee_problem_obj}
\\
\text{s.t.}\quad
&
\mathrm{Tr}\!\left(\boldsymbol{\Psi}\mathbf{V}\mathbf{V}^H\boldsymbol{\Psi}^H\right)\le P_{\max},
\label{eq:ee_problem_power}
\\
&
\theta_{\ell,m}\in[0,2\pi), \qquad \forall \ell\in\mathcal{L},\, m\in\mathcal{M},
\label{eq:ee_problem_phase}
\\
&
\mathbf{q} \in \mathcal{Q}, \label{eq:ee_problem_q}
\end{align}
\end{subequations}
where $\mathcal{Q}
=
\left\{
\mathbf{q}\in\mathbb{R}^3:
l_{u,\min}\le l_u\le l_{u,\max},
\right\}$, denotes a feasible region for UAV placement with $l \in \{x,y,z\}$.
Constraint~\eqref{eq:ee_problem_power} limits the maximum transmit power available at the SIM output, ensuring compliance with practical hardware and regulatory constraints. Constraint~\eqref{eq:ee_problem_phase} enforces the unit-modulus property of the SIM elements, reflecting the phase-only control capability of passive meta-atoms. Finally, constraint \eqref{eq:ee_problem_q} restricts the UAV position within a predefined three-dimensional region, accounting for operational limits such as altitude regulations and coverage requirements.

Problem~\eqref{eq:ee_problem} is highly non-convex and challenging to solve due to several coupled factors. First, the objective function is fractional, involving the ratio of a non-concave sum-rate function and a non-linear power consumption model. Second, the cascaded structure of the SIM introduces a multiplicative coupling across multiple layers, making the optimization variables strongly interdependent. Third, the wireless channel depends nonlinearly on the UAV position $\mathbf{q}$ through both distance-dependent path loss and angular variations. Finally, the unit-modulus constraints on the SIM coefficients render the feasible set non-convex. These challenges necessitate the development of efficient decomposition and approximation techniques, which will be addressed in the subsequent sections.

  \section{Proposed EE Optimization Framework} \label{sec3}

To solve Problem~\eqref{eq:ee_problem}, we develop a transform-based alternating optimization framework that integrates fractional programming, auxiliary-variable transformations, Riemannian optimization, and SCA. The main challenge arises from the fractional structure of the objective function, the strong coupling among optimization variables, and the non-convex constraints imposed by the SIM architecture and UAV deployment. The proposed approach systematically decomposes the original problem into tractable subproblems that can be solved efficiently in an iterative manner.

\subsection{Fractional Programming Reformulation}

We first address the fractional objective using Dinkelbach’s method \cite{dinkelbach1967nonlinear}. Introducing an auxiliary parameter $\eta \geq 0$, define
\begin{subequations}
\begin{equation}\label{eq:f_eta}
\mathcal{F}(\eta)
=
\max_{\mathbf{V},\,\{\mathbf{\Phi}_\ell\},\,\mathbf{q}}
\;
B R(\mathbf{V},\{\mathbf{\Phi}_\ell\},\mathbf{q})
-\eta P_{\rm tot}(\mathbf{V},\{\mathbf{\Phi}_\ell\}),
\end{equation}
    \begin{equation}
    \text{s.t.}\quad  \eqref{eq:ee_problem_power}-\eqref{eq:ee_problem_q}.
\end{equation}
\end{subequations} 

\begin{theorem}
Assume $R(\cdot)\ge 0$ and $P_{\rm tot}(\cdot)>0$ over the feasible set. Then, the optimal EE $\eta$ satisfies $\mathcal{F}(\eta)=0$, and can be obtained via
\begin{equation}\label{eq:dinkelbach_update}
\eta^{(t+1)}=
\frac{
B R(\mathbf{V}^{(t)},\{\mathbf{\Phi}_\ell^{(t)}\},\mathbf{q}^{(t)})
}{P_{\rm tot}(\mathbf{V}^{(t)},\{\mathbf{\Phi}_\ell^{(t)}\})}.
\end{equation}
\end{theorem}

\begin{proof}
The proof follows directly from fractional programming theory. The function $\mathcal{F}(\eta)$ is continuous and strictly decreasing in $\eta$, and the unique root satisfies $\mathcal{F}(\eta)=0$. The iterative update converges to $\eta$ \cite{dinkelbach1967nonlinear}.
\end{proof}

For a fixed $\eta$, we consider the subtractive problem
\begin{subequations}
    \begin{equation}\label{eq:subtractive_problem}
\max_{\mathbf{V},\,\{\mathbf{\Phi}_\ell\},\,\mathbf{q}}
\;
B R(\mathbf{V},\{\mathbf{\Phi}_\ell\},\mathbf{q})
-\eta P_{\rm tot}(\mathbf{V},\{\mathbf{\Phi}_\ell\}),
\end{equation}
\begin{equation}
    \text{s.t.}\quad  \eqref{eq:ee_problem_power}-\eqref{eq:ee_problem_q}.
\end{equation}
\end{subequations}
Define
$
z_{k,i}\triangleq \mathbf{h}_k^H(\mathbf{q})\boldsymbol{\Psi}\mathbf{v}_i,
$
so that
\begin{equation}
\gamma_k=
\frac{|z_{k,k}|^2}{\sum_{i\neq k}|z_{k,i}|^2+\sigma_k^2}.
\end{equation}
The objective becomes
\begin{equation}\label{eq:subtractive_rate}
G(\mathbf{V},\{\mathbf{\Phi}_\ell\},\mathbf{q};\eta)
=
\sum_{k=1}^{K} B\log_2(1+\gamma_k)-\eta P_{\rm tot}.
\end{equation}

The non-convexity of $\log(1+\gamma_k)$ motivates the use of auxiliary-variable transforms \cite{shen2018fractional,shi2011iteratively}.

 \begin{lemma}
For any $\gamma_k \ge 0$, the following holds:
\begin{equation}
\log_2(1+\gamma_k)
=
\max_{\alpha_k \ge 0}
\left[
\log_2(1+\alpha_k)
-\frac{\alpha_k}{\ln 2}
+\frac{1+\alpha_k}{\ln 2}\,\zeta_k
\right],
\end{equation}
where
\begin{equation}
\zeta_k \triangleq \frac{\gamma_k}{1+\gamma_k}=
\frac{|z_{k,k}|^2}{\sum_{i=1}^{K}|z_{k,i}|^2+\sigma_k^2},
\end{equation}
and the optimum is attained at $\alpha_k = \gamma_k$.
\end{lemma}

\begin{proof}
The result follows from the convex conjugate representation of the logarithmic function. By introducing the auxiliary variable $\alpha_k$, the objective becomes a tight lower bound of $\log_2(1+\gamma_k)$, which is maximized at $\alpha_k = \gamma_k$ \cite{shen2018fractional}.
\end{proof}

\begin{lemma}
The parameter $\zeta_k$ admits the following equivalent reformulation:
\begin{equation}
\zeta_k
=
\max_{y_k \in \mathbb{C}}
\;
2\Re\{y_k^* z_{k,k}\}
-
|y_k|^2\!\left(\sum_{i=1}^{K}|z_{k,i}|^2+\sigma_k^2\right),
\end{equation}
with the optimal solution
\begin{equation}
y_k=
\frac{z_{k,k}}{\sum_{i=1}^{K}|z_{k,i}|^2+\sigma_k^2}.
\end{equation}
\end{lemma}

\begin{proof}
The right-hand side is a concave quadratic function of $y_k$. Taking the derivative with respect to $y_k^*$ and setting it to zero yields the stated optimal solution. Substituting this solution back into the objective recovers $\zeta_k$, proving equivalence \cite{shen2018fractional}.
\end{proof}

Applying the above transformations yields the following equivalent reformulation.

\begin{theorem}
For any fixed $\eta$, Problem~\eqref{eq:subtractive_problem} is equivalent to
\begin{subequations}
\begin{equation}\label{eq:fp_problem}
\max_{\substack{\mathbf{V},\,\{\mathbf{\Phi}_\ell\},\,\mathbf{q}\\ \{\alpha_k \ge 0,\, y_k\}}}
\;
J(\mathbf{V},\{\mathbf{\Phi}_\ell\},\mathbf{q},\{\alpha_k,y_k\};\eta),
\end{equation}
\begin{equation}
\text{s.t.}\quad \eqref{eq:ee_problem_power}-\eqref{eq:ee_problem_q},
\end{equation}
\end{subequations}
where
\begin{equation}
J=
\sum_{k=1}^{K}
\frac{B}{\ln 2}
\left[
\ln(1+\alpha_k)-\alpha_k
+(1+\alpha_k)\zeta_k
\right]
-\eta P_{\rm tot},
\end{equation}
and
\begin{equation}
\zeta_k
=
2\Re\{y_k^* z_{k,k}\}
-
|y_k|^2\!\left(\sum_{i=1}^{K}|z_{k,i}|^2+\sigma_k^2\right).
\end{equation}
\end{theorem}

\begin{proof}
The proof follows by applying the dual transform to each logarithmic term and the quadratic transform to each SINR-related term $\zeta_k$. Since both transformations are tight, the reformulated problem is equivalent to the original one for fixed $\eta$ \cite{shen2018fractional}.
\end{proof}

Furthermore, for any feasible $(\mathbf{V},\{\mathbf{\Phi}_\ell\},\mathbf{q})$, we have
\begin{equation}
G(\mathbf{V},\{\mathbf{\Phi}_\ell\},\mathbf{q};\eta)
=
\max_{\{\alpha_k \ge 0,\, y_k\}}
J(\mathbf{V},\{\mathbf{\Phi}_\ell\},\mathbf{q},\{\alpha_k,y_k\};\eta).
\end{equation}
The optimal auxiliary variables are given by
\begin{equation}\label{eq:aux_alpha}
\alpha_k = \gamma_k,\qquad
y_k=
\frac{z_{k,k}}{\sum_{i=1}^{K}|z_{k,i}|^2+\sigma_k^2}, \ \forall k.
\end{equation}
This reformulation preserves optimality while yielding a structure that is amenable to efficient alternating optimization.

  \subsection{Digital Beamformer Update}

For fixed SIM phase shifts $\{\mathbf{\Phi}_\ell\}$, UAV position $\mathbf{q}$, and auxiliary variables $\{\alpha_k\}$ and $\{y_k\}$, the objective function in \eqref{eq:fp_problem} becomes a concave quadratic function with respect to the digital beamforming matrix $\mathbf{V}$. 

To facilitate the derivation, we define the effective channel after SIM processing as
\begin{equation}
\bar{\mathbf{h}}_k^H=\mathbf{h}_k^H(\mathbf{q})\boldsymbol{\Psi},\qquad
\bar{\mathbf{H}}=[\bar{\mathbf{h}}_1,\ldots,\bar{\mathbf{h}}_K]^H\in\mathbb{C}^{K\times N}.
\end{equation}
Moreover, introduce the diagonal matrices
$
\mathbf{A}=\mathrm{diag}(a_1,\ldots,a_K),\;
\mathbf{Y}=\mathrm{diag}(b_1,\ldots,b_K),
$
where
$
a_k=(1+\alpha_k)|y_k|^2,
b_k=(1+\alpha_k)y_k.
$
These definitions compactly capture the contributions of the auxiliary variables and user weights. By discarding constant terms independent of $\mathbf{V}$, the digital beamforming subproblem can be written as
\begin{subequations}
    \begin{align}
\max_{\mathbf{V}}\quad
&
-\frac{B}{\ln2}\mathrm{Tr}\!\big(\mathbf{V}^H\bar{\mathbf{H}}^H\mathbf{A}\bar{\mathbf{H}}\mathbf{V}\big)
+\frac{2B}{\ln2}\Re\!\big\{\mathrm{Tr}(\mathbf{Y}^H\bar{\mathbf{H}}\mathbf{V})\big\}
\nonumber\\
&
-\eta\xi\,\mathrm{Tr}\!\big(\boldsymbol{\Psi}\mathbf{V}\mathbf{V}^H\boldsymbol{\Psi}^H\big)
\label{eq:fp_v_subproblem}
\end{align}
\begin{equation}
    \text{s.t.} \quad \eqref{eq:ee_problem_power}
\end{equation}
\end{subequations}
The above problem is a convex quadratic program since the objective is concave in $\mathbf{V}$ and the constraint is convex.

\begin{theorem}
For fixed SIM response $\{\mathbf{\Phi}_\ell\}$, UAV position $\mathbf{q}$, and auxiliary variables $\{\alpha_k\}$, and $\{y_k\}$, the optimal solution to \eqref{eq:fp_v_subproblem} is given by
\begin{equation}\label{eq:fp_v_opt}
\mathbf{V}(\nu)
=
\left(
\frac{B}{\ln2}\bar{\mathbf{H}}^H\mathbf{A}\bar{\mathbf{H}}
+
(\eta\xi+\nu)\boldsymbol{\Psi}^H\boldsymbol{\Psi}
\right)^{\dagger}
\frac{B}{\ln2}\bar{\mathbf{H}}^H\mathbf{Y},
\end{equation}
where $\nu\ge 0$ is the Lagrange multiplier associated with the power constraint \eqref{eq:ee_problem_power}.
\end{theorem}

\begin{proof}
The problem in \eqref{eq:fp_v_subproblem} is a convex quadratic program. Introducing the Lagrangian with multiplier $\nu \ge 0$, we obtain
\begin{align}
\mathcal{L}(\mathbf{V},\nu)
&=
-\frac{B}{\ln2}\mathrm{Tr}\!\big(\mathbf{V}^H\bar{\mathbf{H}}^H\mathbf{A}\bar{\mathbf{H}}\mathbf{V}\big)
+\frac{2B}{\ln2}\Re\!\big\{\mathrm{Tr}(\mathbf{Y}^H\bar{\mathbf{H}}\mathbf{V})\big\}
\nonumber\\
&\quad
-(\eta\xi+\nu)\mathrm{Tr}\!\big(\boldsymbol{\Psi}\mathbf{V}\mathbf{V}^H\boldsymbol{\Psi}^H\big)
+\nu P_{\max}.
\end{align}
Taking the derivative with respect to $\mathbf{V}^H$ and setting it to zero yields the following KKT condition
\begin{equation}
\frac{B}{\ln2}\bar{\mathbf{H}}^H\mathbf{A}\bar{\mathbf{H}}\mathbf{V}
+
(\eta\xi+\nu)\boldsymbol{\Psi}^H\boldsymbol{\Psi}\mathbf{V}
=
\frac{B}{\ln2}\bar{\mathbf{H}}^H\mathbf{Y}.
\end{equation}
Solving this linear matrix equation yields \eqref{eq:fp_v_opt}. The pseudoinverse ensures existence of a solution even in rank-deficient cases.
\end{proof}
The dual variable $\nu$ is selected to satisfy the power constraint
\begin{equation}
\mathrm{Tr}\!\left(\boldsymbol{\Psi}\mathbf{V}(\nu)\mathbf{V}(\nu)^H\boldsymbol{\Psi}^H\right)\le P_{\max}.
\end{equation}
Since the transmit power is monotonically decreasing in $\nu$, the optimal value can be efficiently obtained via a one-dimensional bisection search. If the constraint is inactive, then $\nu=0$; otherwise, $\nu>0$ is adjusted until the constraint is met with equality.


 \subsection{SIM Optimization via Riemannian Manifold Ascent}
Assuming the other variables fixed, the optimization of the SIM phase shifts is subject to unit-modulus constraints, which naturally induce a Riemannian manifold structure \cite{absil2009optimization}. Specifically, for each layer $\ell$, the feasible set is given by
\begin{equation}
\mathcal{M}_{\ell}
=
\{\boldsymbol{\phi}_{\ell}\in\mathbb{C}^{M}:\ |\phi_{\ell,m}|=1,\ \forall m\},
\end{equation}
which corresponds to a product manifold of $M$ complex unit circles. Due to the non-convexity of this constraint set, we adopt Riemannian optimization techniques \cite{absil2009optimization} to efficiently update the SIM coefficients.

For fixed $\mathbf{V}$, $\mathbf{q}$, $\{\alpha_k\}$, $\{y_k\}$, and $\{\mathbf{\Phi}_j\}_{j\neq \ell}$, the objective function with respect to $\boldsymbol{\phi}_{\ell}$ for the $l$-th layers is given by
\begin{equation} \label{eq:sim_obj_expanded}
    \begin{aligned}
        f_{\ell}(\boldsymbol{\phi}_{\ell})
&=
\frac{B}{\ln2}\sum_{k=1}^{K}(1+\alpha_k)
\Big[
2\Re\{y_k^* z_{k,k}(\boldsymbol{\phi}_{\ell})\}
\\&\quad-
|y_k|^2\sum_{i=1}^{K}|z_{k,i}(\boldsymbol{\phi}_{\ell})|^2
\Big]
-\eta\xi\,P_{{\rm tx},\ell}(\boldsymbol{\phi}_{\ell}),
    \end{aligned}
\end{equation}
where $P_{{\rm tx},\ell}$ denotes the transmit power contribution associated with layer $\ell$. The resulting SIM optimization subproblem is
\begin{equation}\label{eq:sim_manifold_subproblem}
\max_{\boldsymbol{\phi}_{\ell}\in\mathcal{M}_{\ell}} \; f_{\ell}(\boldsymbol{\phi}_{\ell}), \quad \text{s.t.} \quad \eqref{eq:ee_problem_phase}
\end{equation}
To facilitate efficient optimization, we exploit the cascaded structure of the SIM. Define the partial matrices
\begin{equation}
\mathbf{A}_{\ell}
=
\mathbf{\Phi}_L\mathbf{W}_L\cdots\mathbf{\Phi}_{\ell+1}\mathbf{W}_{\ell+1},\;
\mathbf{B}_{\ell}
=
\mathbf{W}_{\ell}\mathbf{\Phi}_{\ell-1}\mathbf{W}_{\ell-1}\cdots\mathbf{\Phi}_1\mathbf{W}_1\mathbf{V},
\end{equation}
so that $\boldsymbol{\Psi}\mathbf{V}=\mathbf{A}_{\ell}\mathbf{\Phi}_{\ell}\mathbf{B}_{\ell}$. Using this decomposition, the effective signal term can be expressed in linear form as
\begin{equation}
z_{k,i}=\boldsymbol{\phi}_{\ell}^{T}\mathbf{c}_{k,i}^{(\ell)},
\end{equation}
where $\mathbf{c}_{k,i}^{(\ell)}\in\mathbb{C}^{M}$ is given by
\begin{equation}
\mathbf{c}_{k,i}^{(\ell)}
=
\left(\mathbf{h}_k^H(\mathbf{q})\mathbf{A}_{\ell}\right)^T
\odot
\left(\mathbf{B}_{\ell}\mathbf{e}_i\right),
\end{equation}
where $\odot$ denotes the Hadamard product and $\mathbf{e}_i$ denotes the $i$-th canonical basis vector. Equivalently, its $m$-th element is
\begin{equation}
c_{k,i,m}^{(\ell)}
=
\left[\mathbf{h}_k^H(\mathbf{q})\mathbf{A}_{\ell}\right]_m
\left[\mathbf{B}_{\ell}\mathbf{e}_i\right]_m .
\end{equation} This follows directly by expanding the diagonal structure of $\mathbf{\Phi}_{\ell}$ and collecting the coefficients associated with each meta-atom.
Accordingly, the quadratic terms admit the representation
\begin{equation}
|z_{k,i}|^2=\boldsymbol{\phi}_{\ell}^H \mathbf{C}_{k,i}^{(\ell)} \boldsymbol{\phi}_{\ell},
\quad
\mathbf{C}_{k,i}^{(\ell)}=\mathbf{c}_{k,i}^{(\ell)*}\mathbf{c}_{k,i}^{(\ell)T},
\end{equation}
and the transmit power can be written as
\begin{equation}
P_{{\rm tx},\ell}(\boldsymbol{\phi}_{\ell})
=
\boldsymbol{\phi}_{\ell}^H \mathbf{Q}_{\ell} \boldsymbol{\phi}_{\ell}, \quad \mathbf{Q}_{\ell}\succeq 0.
\end{equation}
The matrix $\mathbf{Q}_{\ell}$ is obtained by expanding
$
P_{{\rm tx},\ell}
=
\|\mathbf{A}_{\ell}\mathbf{\Phi}_{\ell}\mathbf{B}_{\ell}\|_F^2
$
with respect to $\boldsymbol{\phi}_{\ell}$. Specifically, let
$
\mathbf{G}_{\ell}=\mathbf{A}_{\ell}^H\mathbf{A}_{\ell}
$
and
$
\mathbf{S}_{\ell}=\mathbf{B}_{\ell}\mathbf{B}_{\ell}^H.
$
Then,
\begin{equation}
\mathbf{Q}_{\ell}
=
\mathbf{G}_{\ell}\odot \mathbf{S}_{\ell}^{T}
=
\left(\mathbf{A}_{\ell}^H\mathbf{A}_{\ell}\right)
\odot
\left(\mathbf{B}_{\ell}\mathbf{B}_{\ell}^H\right)^T .
\end{equation}
Since $\mathbf{G}_{\ell}\succeq \mathbf{0}$ and $\mathbf{S}_{\ell}^{T}\succeq \mathbf{0}$, it follows from the Schur product theorem that $\mathbf{Q}_{\ell}\succeq \mathbf{0}$ \cite{paulsen1989schur}.
Substituting these expressions into \eqref{eq:sim_obj_expanded} reveals that $f_\ell(\boldsymbol{\phi}_{\ell})$ is a quadratic function of $\boldsymbol{\phi}_{\ell}$.

Using Wirtinger calculus, the Euclidean gradient of $f_\ell$ with respect to $\boldsymbol{\phi}_{\ell}^H$ is obtained as
\begin{equation}\label{eq:fp_euclid_grad_sim}
\nabla_{\boldsymbol{\phi}_{\ell}^{H}}f_{\ell}
=
\frac{B}{\ln2}\sum_{k=1}^{K}(1+\alpha_k)
\left(
y_k^*\,\mathbf{c}_{k,k}^{(\ell)*}
-
|y_k|^2 \sum_{i=1}^{K}\mathbf{C}_{k,i}^{(\ell)}\boldsymbol{\phi}_{\ell}
\right)
-\eta\xi\,\mathbf{Q}_{\ell}\boldsymbol{\phi}_{\ell},
\end{equation}
which follows from standard differentiation of quadratic forms. However, since $\boldsymbol{\phi}_{\ell}$ lies on the manifold $\mathcal{M}_{\ell}$, the update direction must belong to the corresponding tangent space \cite{zhong2024ris}
\begin{equation}
T_{\boldsymbol{\phi}_{\ell}}\mathcal{M}_{\ell}
=
\{\mathbf{z}\in\mathbb{C}^{M}:\ \Re\{\mathbf{z}\odot\boldsymbol{\phi}_{\ell}^{*}\}=0\}.
\end{equation}
Therefore, the Riemannian gradient is obtained by projecting the Euclidean gradient onto the tangent space as \cite{absil2009optimization}
\begin{equation}\label{eq:fp_riem_grad}
\mathrm{grad}\,f_{\ell}
=
\nabla_{\boldsymbol{\phi}_{\ell}^{*}}f_{\ell}
-
\Re\{\nabla_{\boldsymbol{\phi}_{\ell}^{*}}f_{\ell}\odot\boldsymbol{\phi}_{\ell}^{*}\}\odot\boldsymbol{\phi}_{\ell}.
\end{equation}
Starting from an initial feasible point $\boldsymbol{\phi}_{\ell}^{(r)}$, we perform a Riemannian gradient ascent step followed by a retraction onto the manifold to enforce the unit-modulus constraint. The retraction operator is given as
\begin{equation}
\mathcal{R}_{\boldsymbol{\phi}_{\ell}}(\mathbf{d})
=
e^{j\angle(\boldsymbol{\phi}_{\ell}+\mathbf{d})},
\end{equation}
which normalizes each element to unit magnitude. The update rule is then given by
\begin{equation}
\boldsymbol{\phi}_{\ell}^{(r+1)}
=
\mathcal{R}_{\boldsymbol{\phi}_{\ell}^{(r)}}\!\left(\alpha_r\,\mathrm{grad}\,f_{\ell}(\boldsymbol{\phi}_{\ell}^{(r)})\right),
\end{equation}
where $\alpha_r>0$ is selected using Armijo backtracking line search \cite{karthiga2026multistart} to ensure monotonic improvement of the objective. By sequentially updating all layers $\ell=1,\ldots,L$, one complete SIM optimization cycle is obtained. This procedure guarantees non-decreasing objective values at each iteration and efficiently exploits the manifold structure of the problem. Overall optimization procedure to optimize the multi-layer SIM response is formally given in Algorithm 1.

  \begin{algorithm}[t]
\caption{SIM Optimization via Riemannian Manifold Ascent}
\label{alg:sim_update}
\begin{algorithmic}[1]
\Require $\mathbf{V}$, $\mathbf{q}$, $\{\alpha_k\}$, $\{y_k\}$, initial $\{\boldsymbol{\phi}_\ell^{(0)}\}$, tolerance $\epsilon_{\rm sim}$, maximum iterations $R_{\max}$
\Ensure Updated $\{\boldsymbol{\phi}_\ell\}$
\State Initialize iteration index $r \leftarrow 0$
\Repeat
    \For{$\ell = 1$ to $L$}
        \State Construct $\mathbf{A}_{\ell}$ and $\mathbf{B}_{\ell}$
        \State Compute $\{\mathbf{c}_{k,i}^{(\ell)}, \mathbf{C}_{k,i}^{(\ell)}\}$ and $\mathbf{Q}_{\ell}$
        \State Evaluate Euclidean gradient $\nabla_{\boldsymbol{\phi}_\ell^H} f_\ell$ using \eqref{eq:fp_euclid_grad_sim}
        \State Compute Riemannian gradient $\mathrm{grad}\, f_\ell$ using \eqref{eq:fp_riem_grad}
        \State Determine step size $\alpha_r$ via Armijo backtracking line search
        \State Update
        \[
        \boldsymbol{\phi}_{\ell}^{(r+1)} =
        \mathcal{R}_{\boldsymbol{\phi}_{\ell}^{(r)}}\!\left(\alpha_r\, \mathrm{grad}\, f_\ell(\boldsymbol{\phi}_{\ell}^{(r)})\right)
        \]
    \EndFor
    \State $r \leftarrow r + 1$
\Until{$\max_{\ell}\|\boldsymbol{\phi}_{\ell}^{(r)} - \boldsymbol{\phi}_{\ell}^{(r-1)}\| \le \epsilon_{\rm sim}$}
\end{algorithmic}
\end{algorithm}

\subsection{UAV 3D Position Optimization}

For fixed digital beamformer $\mathbf{V}$, SIM phase matrices $\{\mathbf{\Phi}_\ell\}$, and auxiliary variables $\{\alpha_k\}$ and $\{y_k\}$, the UAV position only affects the transformed objective through the user channels. Hence, the UAV position update reduces to maximizing the function
\begin{equation}\label{eq:q_sca_surrogate}
f_q(\mathbf{q})
=
\sum_{k=1}^{K}
\frac{B}{\ln2}(1+\alpha_k)
\left[
2\Re\{y_k^* z_{k,k}(\mathbf{q})\}
-
|y_k|^2\sum_{i=1}^{K}|z_{k,i}(\mathbf{q})|^2
\right],
\end{equation}
under the constraint \eqref{eq:ee_problem_q}, where $z_{k,i}(\mathbf{q})=\mathbf{h}_k^H(\mathbf{q})\boldsymbol{\Psi}\mathbf{v}_i$ and we define $\mathbf{f}_i=\boldsymbol{\Psi}\mathbf{v}_i$ for notational simplicity. The circuit power and SIM-dependent transmit power are independent of $\mathbf{q}$ for fixed $\mathbf{V}$ and $\{\mathbf{\Phi}_\ell\}$, and thus do not affect this subproblem. The gradient of $f_q(\mathbf{q})$ can be obtained using the chain rule as
\begin{equation}\label{eq:fp_q_grad_updated}
\nabla_{\mathbf{q}}f_q
=
\frac{B}{\ln2}\sum_{k=1}^{K}(1+\alpha_k)
\left(
2\Re\{y_k^*\nabla_{\mathbf{q}}z_{k,k}\}
-
|y_k|^2\sum_{i=1}^{K}\nabla_{\mathbf{q}}|z_{k,i}|^2
\right),
\end{equation}
where
\[
\nabla_{\mathbf{q}}|z_{k,i}|^2
=
2\Re\{z_{k,i}^*\nabla_{\mathbf{q}}z_{k,i}\},
\quad
\nabla_{\mathbf{q}}z_{k,i}
=
\sum_{m=1}^{M}[\mathbf{f}_i]_m\nabla_{\mathbf{q}}h_{k,m}^*(\mathbf{q}).
\]

To explicitly characterize the position dependence, recall that the channel coefficient can be expressed as $h_{k,m}(\mathbf{q})=\sqrt{\beta_k(\mathbf{q})}\,\tilde h_{k,m}(\mathbf{q})$, where $\beta_k(\mathbf{q})=\beta_0 d_k(\mathbf{q})^{-\alpha}$ is the path loss and $d_k(\mathbf{q})=\|\mathbf{q}-\mathbf{u}_k\|_2$ is the distance. Applying the product rule yields
\begin{equation}
\nabla_{\mathbf{q}}h_{k,m}^*(\mathbf{q})
=
\frac{\partial \sqrt{\beta_k(\mathbf{q})}}{\partial \mathbf{q}}\,\tilde h_{k,m}^*(\mathbf{q})
+
\sqrt{\beta_k(\mathbf{q})}\,\nabla_{\mathbf{q}}\tilde h_{k,m}^*(\mathbf{q}).
\end{equation}

The first term captures the variation of large-scale fading and is given by
\begin{equation}
\frac{\partial \sqrt{\beta_k(\mathbf{q})}}{\partial \mathbf{q}}
=
-\frac{\alpha}{2d_k(\mathbf{q})}
\sqrt{\beta_k(\mathbf{q})}
\frac{\mathbf{q}-\mathbf{u}_k}{d_k(\mathbf{q})}.
\end{equation}
The second term accounts for angular variations of the array response. Since the scattered component is independent of $\mathbf{q}$, only the LoS component contributes, yielding
\begin{equation}
\nabla_{\mathbf{q}}\tilde h_{k,m}^*(\mathbf{q})
=
\sqrt{\frac{K_k}{K_k+1}}
\left(
\frac{\partial a_m^*}{\partial \vartheta_k}\nabla_{\mathbf{q}}\vartheta_k
+
\frac{\partial a_m^*}{\partial \varphi_k}\nabla_{\mathbf{q}}\varphi_k
\right).
\end{equation}

Let $\mathbf{r}_k=\mathbf{u}_k-\mathbf{q}$, $\rho_k=\sqrt{r_{k,x}^2+r_{k,y}^2}$, and $d_k=\|\mathbf{r}_k\|_2$. The elevation and azimuth angles are given by
\begin{equation}
   \vartheta_k=\arctan\!\left(\frac{-r_{k,z}}{\rho_k}\right),\qquad
\varphi_k=\operatorname{atan2}(r_{k,y},r_{k,x}),
\end{equation}

with gradients
\begin{equation}
    \nabla_{\mathbf{q}}\vartheta_k=
\begin{bmatrix}
-\dfrac{r_{k,x}r_{k,z}}{\rho_k d_k^2}\\
-\dfrac{r_{k,y}r_{k,z}}{\rho_k d_k^2}\\
\dfrac{\rho_k}{d_k^2}
\end{bmatrix},
\quad
\nabla_{\mathbf{q}}\varphi_k=
\begin{bmatrix}
\dfrac{r_{k,y}}{\rho_k^2}\\
-\dfrac{r_{k,x}}{\rho_k^2}\\
0
\end{bmatrix}.
\end{equation}

Although the gradient can be computed in closed form, the resulting problem remains non-convex due to the nonlinear dependence of the channel on $\mathbf{q}$. To address this issue, we adopt a SCA approach. Specifically, at iteration $t$, we construct a concave surrogate function around the current point $\mathbf{q}^{(t)}$ as
\begin{equation}
\widetilde f_q(\mathbf{q}\mid \mathbf{q}^{(t)})
=
f_q(\mathbf{q}^{(t)})
+
\nabla_{\mathbf{q}}f_q(\mathbf{q}^{(t)})^T(\mathbf{q}-\mathbf{q}^{(t)})
-\frac{\tau_t}{2}\|\mathbf{q}-\mathbf{q}^{(t)}\|_2^2,
\end{equation}
where $\tau_t>0$ is chosen to ensure that $\widetilde f_q$ is a global lower bound of $f_q$ and is tight at $\mathbf{q}^{(t)}$.

The UAV position is constrained within the feasible region
$\mathcal{Q}$ according to the constraint \eqref{eq:ee_problem_q}.
Given the gradient, maximizing the surrogate over $\mathcal{Q}$ yields the update
\begin{equation} \label{eq_spostamento}
\mathbf{q}^{(t+1)}
=
\Pi_{\mathcal{Q}}
\left(
\mathbf{q}^{(t)}
+
\frac{1}{\tau_t}\nabla_{\mathbf{q}}f_q(\mathbf{q}^{(t)})
\right),
\end{equation}
where $\Pi_{\mathcal{Q}}(\cdot)$ denotes Euclidean projection onto $\mathcal{Q}$. This update guarantees monotonic improvement of the objective and preserves feasibility, thereby enabling efficient and stable UAV position optimization.

 \begin{algorithm}[t]
\caption{Joint EE Maximization Framework for SIM-Assisted UAV Communications}
\label{alg:overall_framework}
\begin{algorithmic}[1]
\Require Initial feasible $\mathbf{V}^{(0)}$, $\{\mathbf{\Phi}_\ell^{(0)}\}$, $\mathbf{q}^{(0)}$, tolerances $\epsilon_{\rm out}$ and $\epsilon_{\rm in}$, maximum iterations $T_{\max}$
\Ensure Optimized $\mathbf{V}$, $\{\mathbf{\Phi}_\ell\}$, and $\mathbf{q}$

\State Set $t\leftarrow 0$
\State Initialize $\mathbf{V}\leftarrow\mathbf{V}^{(0)}$, $\{\mathbf{\Phi}_\ell\}\leftarrow\{\mathbf{\Phi}_\ell^{(0)}\}$, and $\mathbf{q}\leftarrow\mathbf{q}^{(0)}$
\State Compute $\eta^{(0)}$ using \eqref{eq:dinkelbach_update}

\Repeat
    \State Update $\{\alpha_k\}$ and $\{y_k\}$ using \eqref{eq:aux_alpha}
    \State Update $\mathbf{V}$ using \eqref{eq:fp_v_opt}, with $\nu$ found by bisection
    \State Update $\{\boldsymbol{\phi}_\ell\}_{\ell=1}^{L}$ using Algorithm~\ref{alg:sim_update}
    \State Set $\mathbf{\Phi}_\ell \leftarrow \mathrm{diag}(\boldsymbol{\phi}_\ell)$, $\forall \ell$
    \State Compute $\nabla_{\mathbf{q}} f_q$ using \eqref{eq:fp_q_grad_updated}
    \State Update $\mathbf{q}$ using \eqref{eq_spostamento}
    \State Update $\eta^{(t+1)}$ using \eqref{eq:dinkelbach_update}
    \State $t \leftarrow t+1$
\Until{$|\eta^{(t)}-\eta^{(t-1)}|\le \epsilon_{\rm out}$ or $t\ge T_{\max}$}

\end{algorithmic}
\end{algorithm}
 
\subsection{Convergence}

The convergence of the proposed algorithm follows from the monotonic behavior of each update block. For a fixed Dinkelbach parameter $\eta$, the transformed objective in \eqref{eq:fp_problem} is optimized in an alternating manner over the auxiliary variables, the digital beamformer, the SIM phase shifts, and the UAV position. The auxiliary variables $\{\alpha_k,y_k\}$ are updated in closed form according to \eqref{eq:aux_alpha}, which gives the exact maximizer of the transformed objective for fixed primal variables. The digital beamformer is then updated using \eqref{eq:fp_v_opt}, which is the global optimizer of the corresponding concave quadratic subproblem under the transmit-power constraint. Therefore, neither of these two steps decreases the transformed objective.

For the SIM phase update, each layer is optimized on the product manifold of complex unit-modulus variables. The Riemannian gradient in \eqref{eq:fp_riem_grad} provides a feasible ascent direction on the manifold, and the Armijo line search selects a step size that guarantees a non-decreasing value of the layer-wise objective. Since the layers are updated sequentially, one complete SIM update cycle also does not decrease the transformed objective. Similarly, the UAV position update is obtained by maximizing the surrogate function in \eqref{eq:q_sca_surrogate}, which is tight at the current point and selected as a lower bound of the original position-dependent objective. Hence, the projected update in \eqref{eq_spostamento} guarantees
\begin{equation}
f_q(\mathbf{q}^{(t+1)}) \ge f_q(\mathbf{q}^{(t)}).
\end{equation}
Combining the above arguments, the inner alternating optimization procedure satisfies
\begin{equation}
J^{(t+1)} \ge J^{(t)},
\end{equation}
where $J^{(t)}$ denotes the transformed objective value at iteration $t$ for a fixed $\eta$.

The generated sequence is also bounded from above. This is because the transmit power is limited by \eqref{eq:ee_problem_power}, the SIM phase shifts lie on compact unit-modulus manifolds, and the UAV position is restricted to the bounded feasible set defined by  \eqref{eq:ee_problem_q}. Therefore, the non-decreasing sequence of transformed objective values converges to a finite limit. After the inner alternating updates converge, the Dinkelbach parameter is updated according to \eqref{eq:dinkelbach_update}. Since the total consumed power is strictly positive over the feasible set, the Dinkelbach procedure drives the residual function $\mathcal{F}(\eta)$ in \eqref{eq:f_eta} toward zero. Consequently, the overall algorithm converges to a stationary solution of the original EE maximization problem. 

A typical convergence behaviour of the proposed design with SIM layers $L=2$, size of each layer $M=36,25$ and with the number of users $K=4,2$ is presented in Fig.~\ref{fig:placeholder}.

\subsection{Complexity Analysis}

We analyze the computational complexity of one complete iteration of the proposed joint optimization algorithm. First, the auxiliary-variable update requires computing the effective terms $z_{k,i}$ for all user--stream pairs, which involves the product $\mathbf{H}^H\boldsymbol{\Psi}\mathbf{V}$ and has complexity $\mathcal{O}(K^2M)$. Once these terms are available, updating $\{\alpha_k,y_k\}$ only requires scalar operations with complexity $\mathcal{O}(K^2)$. Second, the digital beamformer update requires forming the matrices $\bar{\mathbf{H}}^H\mathbf{A}\bar{\mathbf{H}}$ and $\boldsymbol{\Psi}^H\boldsymbol{\Psi}$, followed by solving an $N\times N$ linear system. This step has complexity $\mathcal{O}(K^2N+MN^2+I_{\nu}N^3)$, where $I_{\nu}$ is the number of bisection iterations used to find the dual variable $\nu$. In our case we assume always $N=K$, thus this reduces to $\mathcal{O}(K^3+MK^2+I_{\nu}K^3)$. Third, for the SIM update, each layer requires constructing the layer-dependent quantities and evaluating the Riemannian gradient. The dominant cost comes from the matrix operations involving $M$ SIM elements at each layer and all $K^2$ user--stream pairs (due to summation over all users for each gradient computation), giving complexity $\mathcal{O}(K^2M^2)$ per layer per manifold iteration. Therefore, for $L$ layers and $R_{\rm sim}$ Riemannian iterations, the SIM update requires $\mathcal{O}(R_{\rm sim}LK^2M^2)$. Fourth, the UAV update requires evaluating the channel-gradient terms for all users, streams, and SIM elements, resulting in complexity $\mathcal{O}(I_qK^2M)$, where $I_q$ denotes the number of SCA/backtracking iterations. Combining all steps, the complexity per outer iteration is given as
\begin{equation}
\mathcal{O}\!\left(
I_{\nu}K^3
+
R_{\rm sim}LK^2M^2
+
I_qK^2M
\right).
\end{equation}
For large SIM deployments, the term $\mathcal{O}(R_{\rm sim}LK^2M^2)$ dominates, indicating that the SIM phase optimization is the main computational bottleneck of the proposed framework.
 \begin{figure}
    \centering
    \includegraphics[width=0.8\linewidth]{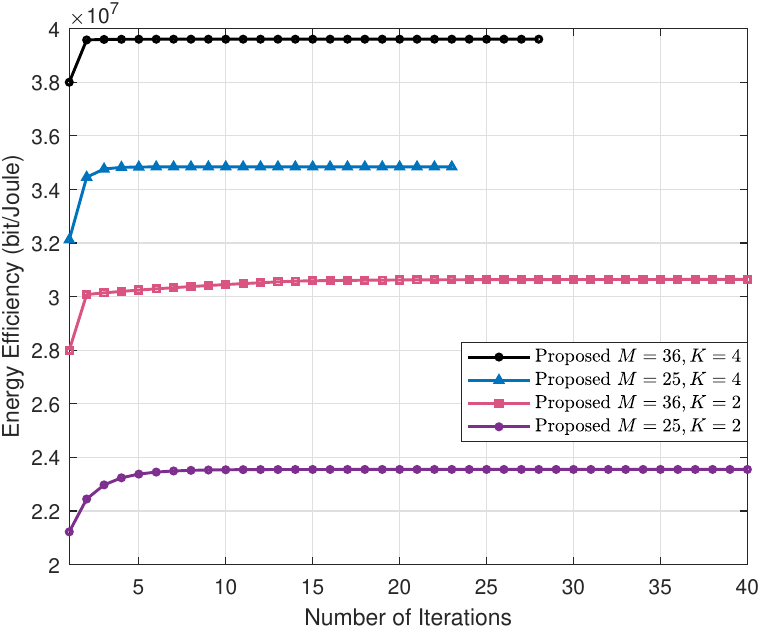}
    \caption{Convergence behaviour of the proposed method with SIM size $M=36,25$, with number of users $K=2, 4$ and number of layers $L=2$.}
    \label{fig:placeholder}
\end{figure}
\section{Numericals Results and Discussion} \label{sec4}
We consider the UAV-assisted multi-user system operating at a carrier frequency of $f_c = 28$ GHz with a bandwidth of $B = 10$ MHz. The speed of light is set to $c = 3 \times 10^8$ m/s, yielding a wavelength $\lambda = c/f_c$ and wavenumber $k_0 = 2\pi/\lambda$. The network consists of $K=2$ or $4$ single-antenna users, uniformly distributed over a square area of $[-30m,30m]$  in both horizontal dimensions, while the UAV position is constrained within $x,y \in [-30m,30m]$ and altitude $h \in [15m,30m]$. A Rician fading channel is assumed with a Rician factor $K_k = 10$ dB. The LoS component follows a distance-dependent path-loss model $
\beta(d) = \beta_0 \left(\frac{d}{d_0}\right)^{-\alpha}$,
where the reference distance is set to $d_0 = 1$ m, and the reference gain is given by 
$\beta_0 = \left(\frac{\lambda}{4\pi d_0}\right)^2$, which corresponds to approximately $-61.38$ dB at $f_c=28$ GHz and the path-loss exponent is set to $\alpha=2.3$.
The SIM architecture employs $M=36$ or $25$ elements arranged in a uniform planar array with half-wavelength spacing $\lambda/2$, and the feed spacing is also $\lambda/2$.  The inter-layer spacing of SIM is set to $d_t = \lambda/8$, and the effective element area is set to $A_t = (\lambda/4)^2$. The thermal noise power is computed using a noise spectral density of $-174$ dBm/Hz.  The power consumption model includes a power amplifier inefficiency factor $\xi_{\text{PA}}=1.1$, RF chain power $P_{\text{RF}}=0.4$ W per chain, baseband power $P_{\text{BB}}=4.5$ W, element-wise power consumption $P_E=0.01$ W per element, and control power $P_{\text{ctrl}}=0.316$ W. In the following, unless otherwise specified, the number of SIM layer is set to $L=2$. Furthermore, the number of RF chains for SIM-assisted UAV system is fixed equal to the number of users. The plotted results are averaged over $50$ independent channel realizations.


\begin{figure*}
    \centering
 \begin{minipage}{0.45\textwidth}
      \centering
    \includegraphics[width=\linewidth]{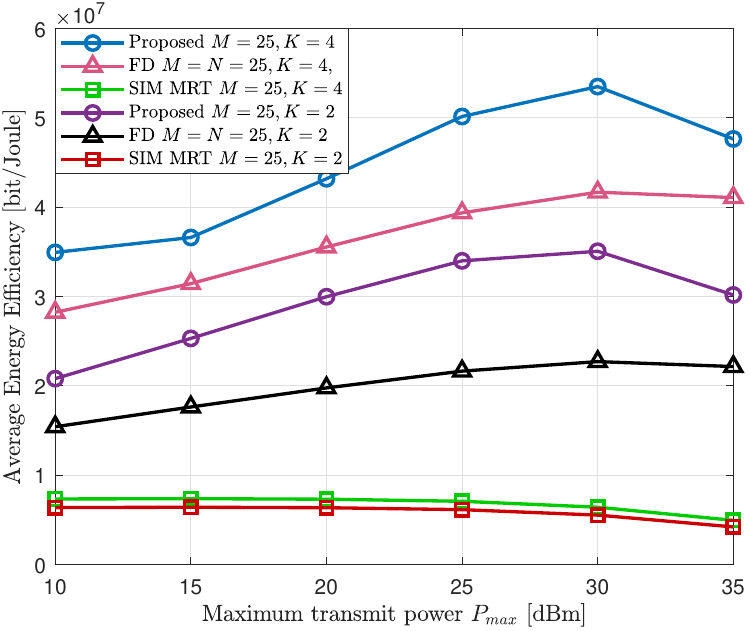}
    \caption{EE as a function of transmit power with layers $L=2$ and $M=25$.}
    \label{fig3}
\end{minipage}  
      \begin{minipage}{0.45\textwidth}
      \centering
    \includegraphics[width=\linewidth]{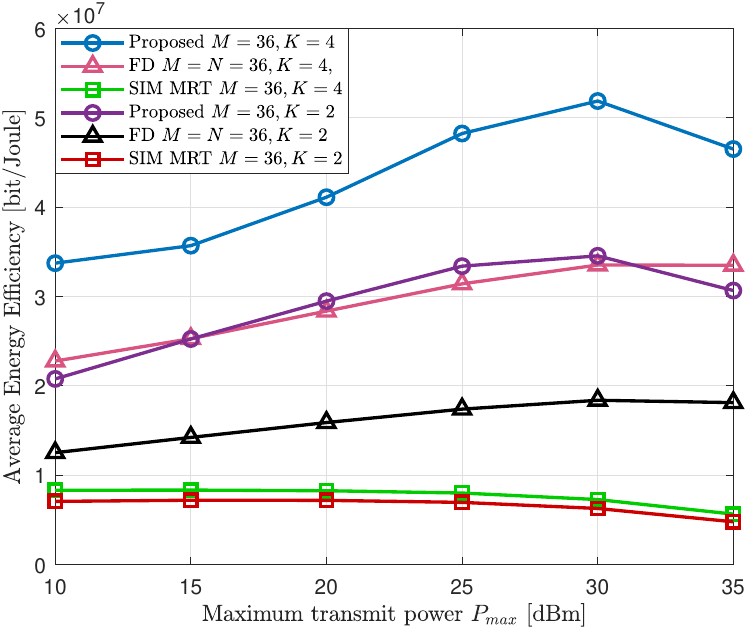}
    \caption{EE as a function of transmit power with layers $L=2$ and $M=36$.}
    \label{fig4}
    \end{minipage}  
\end{figure*}
For performance comparison, we consider the following benchmark schemes.  \emph{1) Fully Digital (FD) Scheme:} In this baseline, a fully digital architecture is assumed where the number of RF chains is equal to the number of transmit antennas, i.e., $N=M$, with $M$ denoting the number of elements per SIM layer. This configuration enables unconstrained digital beamforming without any hardware limitations imposed by the SIM structure. \emph{2) SIM-MRT Scheme:} In this scheme, the SIM structure is retained and jointly optimized with the UAV position, while the digital beamforming is designed according to the MRT principle, given the effective channels including the SIM response.

\begin{figure*}
    \centering
 \begin{minipage}{0.45\textwidth}
      \centering
    \includegraphics[width=\linewidth]{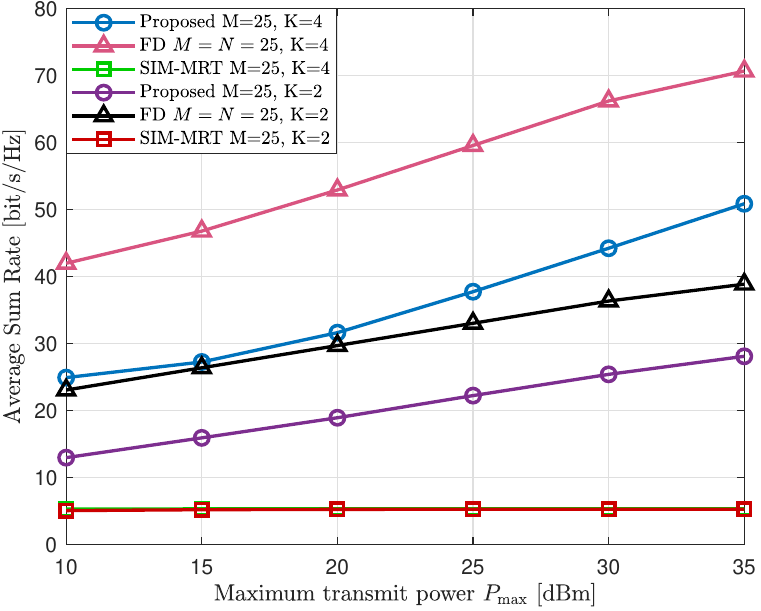}
    \caption{Sum-rate as a function of transmit power with layers $L=2$ and $M=25$.}
    \label{fig5}
\end{minipage}  
      \begin{minipage}{0.45\textwidth}
      \centering
    \includegraphics[width=\linewidth]{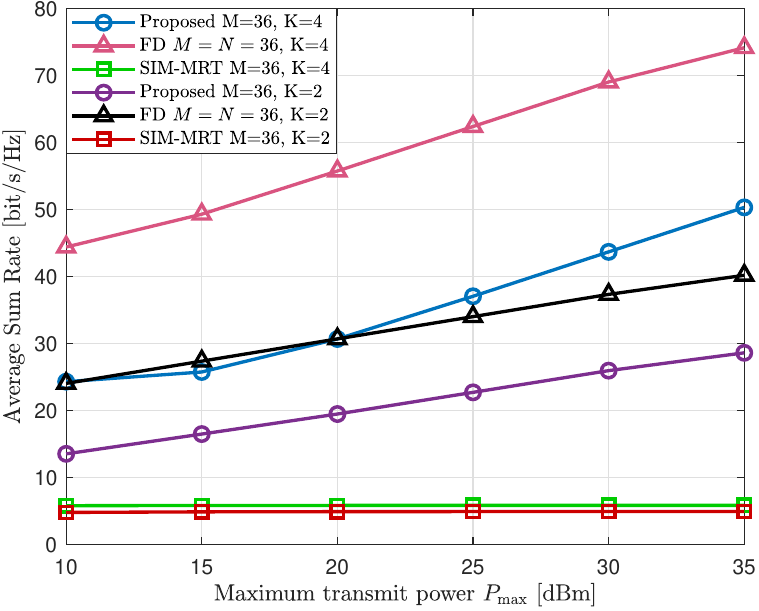 }
    \caption{Sum-rate as a function of transmit power with layers $L=2$ and $M=36$.}
    \label{fig6}
    \end{minipage}  
\end{figure*}

 Figures \ref{fig3} and \ref{fig4} illustrate the average EE versus the maximum transmit power for two SIM configurations with $M=36$ and $M=25$ elements per layer, respectively. Overall, both figures exhibit similar trends, while the larger SIM size yields consistently higher performance due to improved beamforming gain, resulting in greater EE. The proposed scheme significantly outperforms both benchmarks across all transmit power levels and for both $K=4$ and $K=2$ users. For instance, in the case of $M=36$ and $K=4$, the proposed scheme achieves approximately $5.4\times 10^7$ bit/Joule at $P_{\max}=30$ dBm, compared to about $3.4\times 10^7$ bit/Joule for the FD scheme, corresponding to a gain of $\sim 58\%$. For $K=2$, the improvement is also substantial, with the proposed method achieving around $3.5\times 10^7$ bit/Joule versus $2.3\times 10^7$ bit/Joule for FD, yielding an approximate gain of $\sim50\%$. In contrast, the SIM-MRT scheme performs significantly worse, remaining below $1\times 10^7$ bit/Joule, which implies a performance gap exceeding $700\%$ for the proposed design. A similar behavior is observed for $M=25$, although the performance is slightly reduced due to the smaller number of elements. Specifically, for $K=4$, the proposed scheme reaches approximately $5.2\times 10^7$ bit/Joule at $P_{\max}=30$ dBm, while the FD scheme achieves around $4.2\times 10^7$ bit/Joule, corresponding to a gain of about $20\%$–$25\%$. For $K=2$, the gain remains more pronounced, with the proposed scheme outperforming FD by roughly $45\%$–$50\%$. Again, the SIM-MRT baseline remains significantly inferior, with performance values below $0.6\times 10^7$ bit/Joule, leading to an order-of-magnitude gap compared to the proposed scheme.  For both SIM sizes, the EE increases with transmit power up to $30$ dBm and then starts decreasing. This behavior is due to the fact that, beyond this point, the marginal gain in achievable rate becomes limited by inter-user interference and channel saturation, while the total power consumption continues to increase. As a result, the EE, defined as the ratio between rate and power consumption, is no longer improved and begins to degrade.

\begin{figure*}
    \centering
 \begin{minipage}{0.45\textwidth}
      \centering
    \includegraphics[width=\linewidth]{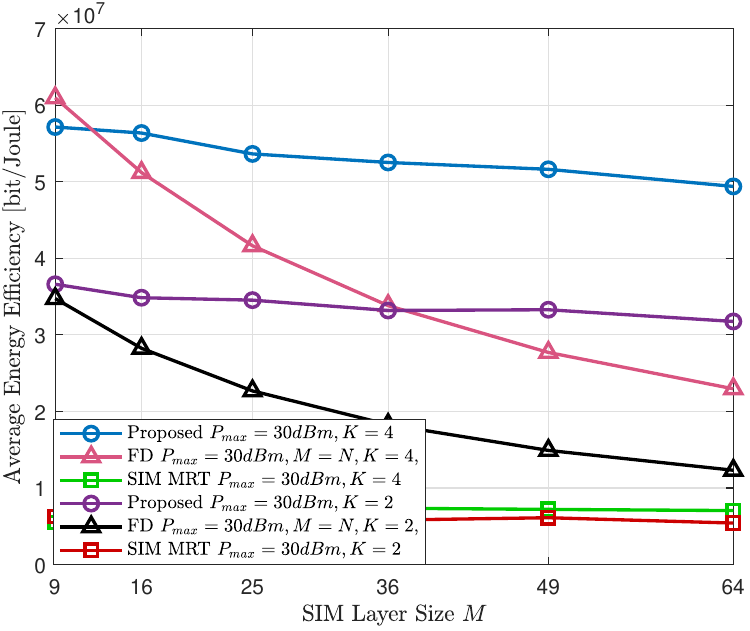}
    \caption{EE as a function of SIM layer size $M$ with $L=2$ and $P_{max}=30$ dBm.}
    \label{fig7}
\end{minipage}  
      \begin{minipage}{0.45\textwidth}
      \centering
    \includegraphics[width=\linewidth]{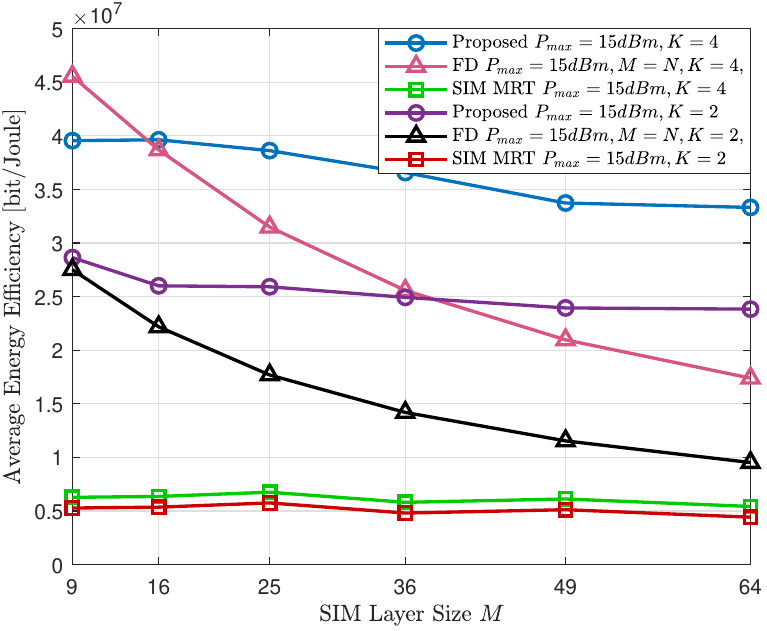}
    \caption{EE as a function of SIM layer size $M$ with $L=2$ and $P_{max}=15$ dBm.}
    \label{fig8}
    \end{minipage}  
\end{figure*}

Figures \ref{fig5} and \ref{fig6} present the average sum rate obtained under the EE maximization framework for both SIM configurations with $M=36$ and $M=25$. In contrast to the EE behavior, the sum rate increases monotonically with the maximum transmit power for all schemes, since higher power directly improves the received signal strength. However, an important observation is that the proposed scheme, while designed to maximize EE, still achieves competitive sum-rate performance and consistently outperforms the SIM-MRT baseline. For instance, in the case of $M=36$ and $K=4$, the proposed method increases from approximately $25$ bit/s/Hz at $P_{\max}=10$ dBm to about $50$ bit/s/Hz at $35$ dBm, whereas the SIM-MRT scheme remains nearly constant around $6$ bit/s/Hz, confirming its inability to exploit additional transmit power due to interference limitations. The FD benchmark achieves the highest sum rate, reaching nearly $75$ bit/s/Hz, as expected from its fully flexible digital beamforming capability. Nevertheless, the performance gap between the proposed and FD schemes is moderate, demonstrating that the proposed SIM-based architecture can achieve a favorable trade-off between hardware complexity and rate performance. It is noteworthy that the MRT performance remains nearly constant because it does not mitigate inter-user interference, causing the SINR—and hence the sum rate—to saturate despite increasing transmit power.
Similar trends are observed for $M=25$, where the proposed scheme attains $\sim51$ bit/s/Hz at $35$ dBm compared to approximately $71$ bit/s/Hz for the FD scheme, corresponding to a gap of about $25\%$–$30\%$, while still significantly outperforming the SIM-MRT scheme by more than $700\%$. Furthermore, increasing the number of users from $K=2$ to $K=4$ leads to substantial sum-rate improvements due to multiplexing gains, which are effectively captured by the proposed design. Overall, these results highlight that maximizing EE does not severely compromise the achievable sum rate, and the proposed method provides a well-balanced performance compared to both high-complexity and low-complexity benchmarks.

 \begin{figure*}
    \centering
 \begin{minipage}{0.45\textwidth}
      \centering
    \includegraphics[width=\linewidth]{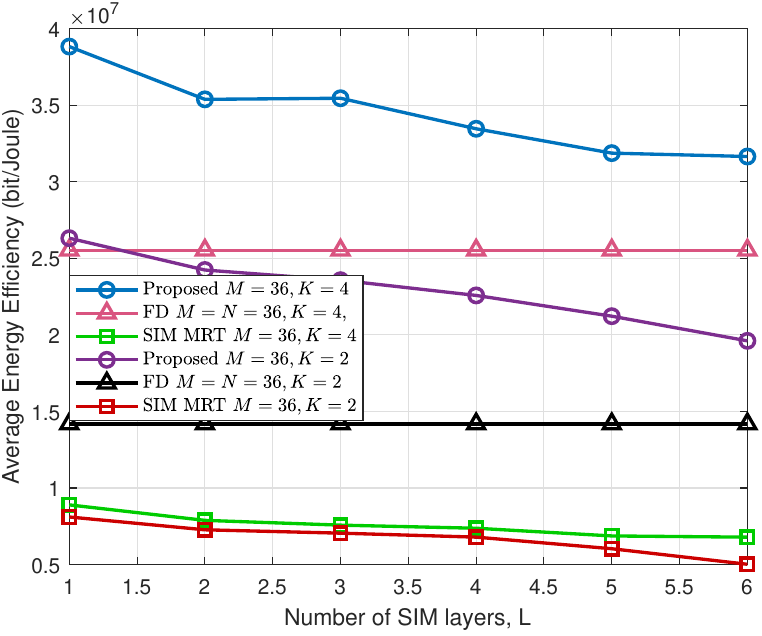}
    \caption{EE as a function of the number of SIM layers $L$, with $M=36$ and $P_{max}=15$ dBm.}
    \label{fig9}
\end{minipage}  
      \begin{minipage}{0.45\textwidth}
      \centering
    \includegraphics[width=\linewidth]{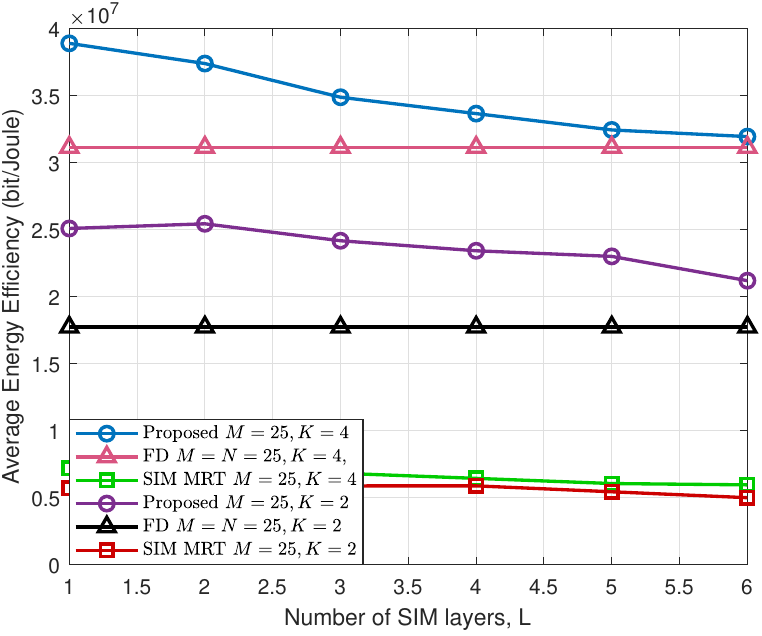}
    \caption{EE as a function of the number of SIM layers $L$, with $M=25$ and $P_{max}=15$ dBm.}
    \label{fig10}
    \end{minipage}  
\end{figure*}

Figures \ref{fig7} and \ref{fig8} illustrate the average EE as a function of the SIM layer size $M$ for two transmit power levels, $P_{\max}=30$ dBm and $P_{\max}=15$ dBm. A common trend observed in both figures is that the proposed scheme maintains a relatively stable performance as $M$ increases, with only a slight decrease in EE. For instance, at $P_{\max}=30$ dBm and $K=4$, the proposed scheme decreases moderately from approximately $5.8\times10^7$ to $5.0\times10^7$ bit/Joule as $M$ increases from $9$ to $64$, indicating that the additional beamforming gain is offset by the increased hardware power consumption. In contrast, the FD benchmark exhibits a sharp degradation with increasing $M$, dropping from $\sim 6.0\times10^7$ to $2.3\times10^7$ bit/Joule for $K=4$, due to the linear increase in RF chain power consumption with $M$. A similar behavior is observed for $K=2$, where the FD scheme decreases more rapidly than the proposed method, highlighting the inefficiency of fully digital architectures at large array sizes under the EE objective. The SIM-MRT scheme remains consistently low and nearly constant across all $M$, confirming that its performance is limited by interference rather than array size. At the lower transmit power of $P_{\max}=15$ dBm, similar trends are observed, albeit at reduced performance levels. The proposed scheme again shows robustness to increasing $M$, with only a gradual decline in EE, while the FD scheme experiences a significant drop as $M$ increases, particularly for $K=2$. For example, for $K=4$, the FD scheme decreases from approximately $4.6\times10^7$ to $1.7\times10^7$ bit/Joule, whereas the proposed method only drops from around $4.0\times10^7$ to $3.4\times10^7$ bit/Joule. These results clearly demonstrate that the proposed SIM-based architecture scales much more favorably with the array size compared to the fully digital baseline. From the results we can conclude that, increasing $M$ does not necessarily improve EE, as the additional circuit power can dominate the rate gains, and the proposed scheme effectively balances this trade-off, achieving superior scalability and robustness compared to both FD and SIM-MRT benchmarks.

Figures \ref{fig9} and \ref{fig10} show the average EE versus the number of SIM layers $L$ for $M=36$ and $M=25$, respectively, at $P_{\max}=15$ dBm. For both SIM sizes, the proposed scheme achieves the highest EE in the multi-user case with $K=4$, confirming the benefit of joint digital beamforming, SIM phase optimization, and UAV placement. However, the EE gradually decreases as $L$ increases. For example, when $M=36$ and $K=4$, the proposed scheme decreases from about $3.9\times 10^7$ bit/Joule at $L=1$ to about $3.2\times 10^7$ bit/Joule at $L=6$, corresponding to an approximate reduction of $18\%$. Similarly, for $M=25$ and $K=4$, the proposed method decreases from about $3.9\times 10^7$ bit/Joule to about $3.2\times 10^7$ bit/Joule, showing a comparable trend. This indicates that adding more SIM layers does not always improve EE, since the additional wave-domain processing gain is eventually outweighed by increased hardware/control power consumption and propagation loss across stacked layers. The FD benchmark remains nearly constant with $L$, as expected, because it does not depend on the SIM layer structure. In contrast, the SIM-MRT scheme remains significantly below the proposed method and exhibits only minor variations, confirming that MRT is limited by inter-user interference and cannot effectively exploit the additional SIM layers. Clearly, the results indicate that a small number of SIM layers is sufficient to obtain most of the energy-efficiency gain, while excessive stacking while improving sum-rate performance may reduce EE due to increased circuit power and diminishing beamforming returns. 

The presented results provide several important insights into the design of SIM-assisted UAV communication systems under an EE objective. First, the proposed joint optimization framework consistently achieves the best performance across all considered scenarios, demonstrating the importance of jointly optimizing the digital precoder, SIM configuration, and UAV position. In particular, the results show that SIM-based architectures can effectively bridge the gap between low-complexity designs and fully digital systems, achieving comparable or even superior EE performance with significantly fewer RF chains. Moreover, increasing the transmit power improves EE only up to a certain point, beyond which the additional power consumption outweighs the marginal rate gains. This highlights the fundamental trade-off inherent in EE maximization and underscores the need for adaptive power utilization rather than simply operating at maximum transmit power.

Furthermore, the impact of system parameters such as the SIM size $M$ and the number of layers $L$ reveals important design guidelines. While increasing the number of elements enhances beamforming capability, the gains in EE are limited due to the associated increase in circuit power, making moderate SIM sizes more favorable. Similarly, adding more SIM layers does not necessarily improve performance, as the additional hardware and propagation losses can offset the benefits of deeper wave-domain processing. In contrast, simple beamforming strategies such as MRT are shown to be highly suboptimal in multi-user scenarios, as they fail to manage inter-user interference and cannot efficiently utilize additional power or structural degrees of freedom. To conclude, these observations confirm that achieving high EE in SIM-assisted UAV systems requires a careful balance between spatial degrees of freedom, hardware complexity, and power consumption, which is effectively captured by the proposed optimization framework.

\section{Conclusions} \label{sec5}

In this paper, we investigated energy-efficient multi-user downlink communications in SIM-equipped UAV systems. By leveraging the wave-domain processing capability of the SIM technology, we developed a hybrid analog and digital transmission framework that significantly reduces RF-chain requirements, while maintaining high beamforming flexibility. We formulated an EE maximization problem that joint optimizes the digital precoder, SIM phase shifts, and UAV 3D position under a realistic hardware-aware power consumption model. To tackle the resulting highly non-convex problem, we proposed a unified optimization framework based on fractional programming, auxiliary variable transformations, to leverage Riemannian manifold optimization, and SCA. The proposed algorithm guarantees monotonic convergence to a stationary solution. Numerical results demonstrated that SIM-equipped UAV systems achieve substantial EE gains over conventional fully digital and MRT-based baselines, while also revealing important trade-offs among transmit power, SIM size, and number of layers.

\ifCLASSOPTIONcaptionsoff
  \newpage
\fi

{\footnotesize
\bibliographystyle{IEEEtran}
\def\baselinestretch{0.9}
\bibliography{main}}

@article{yao2024channel,
  title={Channel estimation for stacked intelligent metasurface-assisted wireless networks},
  author={Yao, Xianghao and An, Jiancheng and Gan, Lu and Di Renzo, Marco and Yuen, Chau},
  journal={IEEE Wireless Commun. Lett.},
  volume={13},
  number={5},
  pages={1349--1353},
  year={2024},
  publisher={IEEE}
}

@article{papazafeiropoulos2025channel,
  title={Channel Estimation for Stacked Intelligent Metasurfaces in Rician Fading Channels},
  author={Papazafeiropoulos, Anastasios and Kourtessis, Pandelis and Kaklamani, Dimitra I and Venieris, Iakovos S},
  journal={IEEE Wireless Commun. Lett.},
  year={2025},
  publisher={IEEE}
}

@article{dong2026deep,
  title={Deep Learning-Based Channel Estimation for Stacked Intelligent Metasurface-Enhanced Multi-User Communications},
  author={Dong, Xin and Chen, Chen and Yu, Gang and Zhou, Lingyou and Yuan, Chenyang and Zhang, Jie},
  journal={IEEE Trans. Veh. Technol.},
  year={2026},
  publisher={IEEE}
}

@article{sheemar2022practical,
  title={Practical hybrid beamforming for millimeter wave massive {MIMO} full duplex with limited dynamic range},
  author={Sheemar, Chandan Kumar and Thomas, Christo Kurisummoottil and Slock, Dirk},
  journal={IEEE Open Journal of the Communications Society},
  volume={3},
  pages={127--143},
  year={2022},
  publisher={IEEE}
}

@article{sheemar2026survey,
  title={A Survey on Stacked Intelligent Metasurfaces: Fundamentals, Recent Advances, and Challenges},
  author={Sheemar, Chandan Kumar and Khan, Wali Ullah and Solanki, Sourabh and Alexandropoulos, George C and Chatzinotas, Symeon},
  journal={arXiv preprint arXiv:2603.05633},
  year={2026}
}

@article{sheemar2026joint,
  title={Joint beamforming and {3D} location optimization for multi-user holographic {UAV} communications},
  author={Sheemar, Chandan Kumar and Mahmood, Asad and Thomas, Christo Kurisummoottil and Alexandropoulos, George C and Querol, Jorge and Chatzinotas, Symeon and Saad, Walid},
  journal={IEEE Transactions on Communications},
  year={2026},
  publisher={IEEE}
}

@article{an2023stacked,
  title={Stacked intelligent metasurfaces for efficient holographic {MIMO} communications in {{6G}}},
  author={An, Jiancheng and Xu, Chao and Ng, Derrick Wing Kwan and Alexandropoulos, George C and Huang, Chongwen and Yuen, Chau and Hanzo, Lajos},
  journal={IEEE J. Sel. Areas Commun.},
  volume={41},
  number={8},
  pages={2380--2396},
  year={2023}
}

@article{an2024stacked_WC,
  title={Stacked intelligent metasurface-aided {MIMO} transceiver design},
  author={An, Jiancheng and Yuen, Chau and Xu, Chao and Li, Hongbin and Ng, Derrick Wing Kwan and Di Renzo, Marco and Debbah, M{\'e}rouane and Hanzo, Lajos},
  journal={IEEE Wirel. Commun.},
  volume={31},
  number={4},
  pages={123--131},
  year={2024},
  publisher={IEEE}
}

@article{lawal2025channel,
  title={Channel estimation for stacked intelligent metasurface-aided network using deep learning},
  author={Lawal, Abdulmajid and Zerguine, Azzedine and Nasir, Ali A and Abed-Meraim, Karim},
  journal={IEEE Commun. Lett.},
  year={2025},
  publisher={IEEE}
}

@article{papazafeiropoulos2025ergodic_outage_3,
  title={Ergodic Mutual Information and Outage Probability for {{SIM}}-Assisted Holographic {{MIMO}} Communications},
  author={Papazafeiropoulos, Anastasios and Kourtessis, Pandelis and Kaklamani, Dimitra I and Venieris, Iakovos S},
  journal={IEEE Trans. Veh. Technol.},
  year={2025},
  publisher={IEEE}
}

@article{papazafeiropoulos2025performance,
  title={Performance of double-stacked intelligent metasurface-assisted multiuser massive {{MIMO}} communications in the wave domain},
  author={Papazafeiropoulos, Anastasios and Kourtessis, Pandelis and Chatzinotas, Symeon and Kaklamani, Dimitra I and Venieris, Iakovos S},
  journal={IEEE Trans. Wireless Commun.},
  year={2025},
  publisher={IEEE}
}

@article{papazafeiropoulos2024achievable,
  title={Achievable rate optimization for large stacked intelligent metasurfaces based on statistical {CSI}},
  author={Papazafeiropoulos, Anastasios and Kourtessis, Pandelis and Chatzinotas, Symeon and Kaklamani, Dimitra I and Venieris, Iakovos S},
  journal={IEEE Wireless Commun. Lett.},
  volume={13},
  number={9},
  pages={2337--2341},
  year={2024},
  publisher={IEEE}
}

@article{papazafeiropoulos2024achievable_2,
  title={Achievable rate optimization for stacked intelligent metasurface-assisted holographic {{MIMO}} communications},
  author={Papazafeiropoulos, Anastasios and An, Jiancheng and Kourtessis, Pandelis and Ratnarajah, Tharmalingam and Chatzinotas, Symeon},
  journal={IEEE Trans. Wireless Commun.},
  volume={23},
  number={10},
  pages={13173--13186},
  year={2024},
  publisher={IEEE}
}

@inproceedings{zarini2025orchestration,
  title={On the Orchestration of {SIM} and {UAV}},
  author={Zarini, Hosein and Kazemi, Seyed Mohsen and An, Jiancheng and Sookhak, Mehdi and Choi, Jinho},
  booktitle={Proc. IEEE ICC},
  pages={2913--2918},
  year={2025}
}

@article{liu2023deployment,
  title={Deployment and robust hybrid beamforming for {UAV} mmWave communications},
  author={Liu, Ke and Liu, Yanming and Yi, Pengfei and Xiao, Zhenyu and Xia, Xiang-Gen},
  journal={IEEE Transactions on Communications},
  volume={71},
  number={5},
  pages={3073--3086},
  year={2023},
  publisher={IEEE}
}

@article{sheemar2025joint,
  title={Joint holographic beamforming and user scheduling with individual {QoS} constraints},
  author={Sheemar, Chandan Kumar and Thomas, Christo Kurisummoottil and Alexandropoulos, George C and Querol, Jorge and Chatzinotas, Symeon and Saad, Walid},
  journal={IEEE Transactions on Vehicular Technology},
  year={2025},
  publisher={IEEE}
}

@article{khan2024reconfigurable,
  title={Reconfigurable intelligent surfaces for {6G} non-terrestrial networks: Assisting connectivity from the sky},
  author={Khan, Wali Ullah and Mahmood, Asad and Sheemar, Chandan Kumar and Lagunas, Eva and Chatzinotas, Symeon and Ottersten, Bj{\"o}rn},
  journal={IEEE Internet of Things Magazine},
  volume={7},
  number={1},
  pages={34--39},
  year={2024},
  publisher={IEEE}
}

@article{li2025ris,
  title={{RIS}-based physical layer security for integrated sensing and communication: A comprehensive survey},
  author={Li, Yongxiao and Khan, Feroz and Ahmed, Manzoor and Soofi, Aized Amin and Khan, Wali Ullah and Sheemar, Chandan Kumar and Asif, Muhammad and Han, Zhu},
  journal={IEEE Internet of Things Journal},
  year={2025},
  publisher={IEEE}
}

@article{liu2021reconfigurable,
  title={Reconfigurable intelligent surfaces: Principles and opportunities},
  author={Liu, Yuanwei and Liu, Xiao and Mu, Xidong and Hou, Tianwei and Xu, Jiaqi and Di Renzo, Marco and Al-Dhahir, Naofal},
  journal={IEEE communications surveys \& tutorials},
  volume={23},
  number={3},
  pages={1546--1577},
  year={2021},
  publisher={IEEE}
}

@article{elmossallamy2020reconfigurable,
  title={Reconfigurable intelligent surfaces for wireless communications: Principles, challenges, and opportunities},
  author={ElMossallamy, Mohamed A and Zhang, Hongliang and Song, Lingyang and Seddik, Karim G and Han, Zhu and Li, Geoffrey Ye},
  journal={IEEE Transactions on Cognitive Communications and Networking},
  volume={6},
  number={3},
  pages={990--1002},
  year={2020},
  publisher={IEEE}
}

@article{huang2019reconfigurable,
  title={Reconfigurable intelligent surfaces for energy efficiency in wireless communication},
  author={Huang, Chongwen and Zappone, Alessio and Alexandropoulos, George C and Debbah, M{\'e}rouane and Yuen, Chau},
  journal={IEEE Transactions on Wireless Communications},
  volume={18},
  number={8},
  pages={4157--4170},
  year={2019},
  publisher={IEEE}
}

@inproceedings{sheemar2023full,
  title={Full-duplex-enabled joint communications and sensing with reconfigurable intelligent surfaces},
  author={Sheemar, Chandan Kumar and Alexandropoulos, George C and Slock, Dirk and Querol, Jorge and Chatzinotas, Symeon},
  booktitle={in Proc. IEEE EUSIPCO},
  pages={1509--1513},
  year={2023}
}

@article{sheemar2024parallel,
  title={Parallel and distributed hybrid beamforming for multicell millimeter wave {MIMO} full duplex},
  author={Sheemar, Chandan Kumar and Chatzinotas, Symeon and Slock, Dirk and Lagunas, Eva and Querol, Jorge},
  journal={IEEE Transactions on Vehicular Technology},
  volume={74},
  number={3},
  pages={4289--4306},
  year={2024},
  publisher={IEEE}
}

@article{fan2025joint,
  title={Joint association and phase shifts design for {UAV}-mounted stacked intelligent metasurfaces-assisted communications},
  author={Fan, Mingzhe and Sun, Geng and Pan, Hongyang and Wang, Jiacheng and An, Jiancheng and Du, Hongyang and Yuen, Chau},
  journal={arXiv preprint arXiv:2508.00616},
  year={2025}
}

@inproceedings{khan2024beyond,
  title={Beyond diagonal {IRS} assisted ultra massive THz systems: A low resolution approach},
  author={Khan, Wali Ullah and Sheemar, Chandan Kumar and Abdullah, Zaid and Lagunas, Eva and Chatzinotas, Symeon},
  booktitle={in Proc. IEEE PIMRC},
  pages={1--5},
  year={2024}
}

@article{sheemar2023irs,
  title={{IRS} assisted {MIMO} full duplex: Rate analysis and beamforming under imperfect CSI},
  author={Sheemar, Chandan Kumar and Solanki, Sourabh and Querol, Jorge and Kumar, Sumit and Chatzinotas, Symeon},
  journal={IEEE Open Journal of the Communications Society},
  volume={4},
  pages={1879--1892},
  year={2023},
  publisher={IEEE}
}

@inproceedings{an2024hybrid,
  title={Hybrid digital-wave domain channel estimator for stacked intelligent metasurface enabled multi-user {MISO} systems},
  author={An, Jiancheng and Chaaban, Anas and others},
  booktitle={Proc. IEEE WCNC},
  pages={1--6},
  year={2024},
  organization={}
}

@article{ginige2025nested,
  title={Nested tensor-based channel estimation for stacked intelligent metasurface-assisted wireless networks},
  author={Ginige, Nipuni and De Sena, Arthur Sousa and Mahmood, Nurul Huda and Di Renzo, Marco and Rajatheva, Nandana and Latva-Aho, Matti},
  journal={IEEE Commun. Lett.},
  year={2025},
  publisher={IEEE}
}

@inproceedings{yao2024sparse,
  title={Sparse channel estimation for stacked intelligent metasurface-assisted mmWave communications},
  author={Yao, Xianghao and An, Jiancheng and Huang, Guojun and Liu, Hao and Gan, Lu and Yuen, Chau},
  booktitle={Proc. IEEE VTS APWCS},
  pages={},
  year={2024},
  address={Singapore},
  organization={}}

@inproceedings{yao2025sparse,
  title={Sparse Bayesian Learning Based Channel Estimation for {SIM}-Assisted Near-Field Communications},
  author={Yao, Xianghao and An, Jiancheng and Gan, Lu and Clerckx, Bruno and Di Renzo, Marco},
  booktitle={Proc. IEEE ICC},
  pages={},
  year={2025},
  address={Montreal, Canada},
  organization={}
}

@article{bahingayi2025refined,
  title={A refined alternating optimization for sum rate maximization in {SIM}-aided multiuser {{MISO}} systems},
  author={Bahingayi, Eduard E and Lin, Shuying and Uysal, Murat and Di Renzo, Marco and Tran, Le-Nam},
  journal={IEEE Wireless Commun. Lett.},
  volume={15},
  pages={1250--1254},
  year={2025},
  publisher={IEEE}
}

@article{papazafeiropoulos2025ergodic_HMIMO,
  title={On the Ergodic Capacity for {SIM}-Aided Holographic {MIMO} Communications},
  author={Papazafeiropoulos, Anastasios and Bartsiokas, Ioannis and Kaklamani, Dimitra I and Venieris, Iakovos S},
  journal={IEEE Wireless Commun. Lett.},
  volume={15},
  pages={1120--1124},
  year={2025},
  publisher={IEEE}
}

@article{darsena2025design,
  title={Design of stacked intelligent metasurfaces with reconfigurable amplitude and phase for multiuser downlink beamforming},
  author={Darsena, Donatella and Verde, Francesco and Iudice, Ivan and Galdi, Vincenzo},
  journal={IEEE Open J. Commun. Soc.},
  year={2025},
  publisher={IEEE}
}

@inproceedings{darsena2024downlink,
  title={Downlink Sum-Rate Maximization for Joint Active and Passive Stacked Intelligent Metasurfaces},
  author={Darsena, Donatella and Verde, Francesco and Iudice, Ivan and Galdi, Vincenzo},
  booktitle={Proc. IEEE VCC},
  pages={},
  address={Virtual Conference},
  year={2024},
  organization={}
}

@article{bahingayi2025scaling,
  title={Scaling achievable rates in {SIM}-aided {MIMO} systems with metasurface layers: A hybrid optimization framework},
  author={Bahingayi, Eduard E and Perovi{\'c}, Nemanja Stefan and Tran, Le-Nam},
  journal={IEEE Wireless Commun. Lett.},
  year={2025},
  publisher={IEEE}
}

@article{xia2025statistical,
  title={Statistical {CSI}-Enabled Optimization for Beyond Diagonal {SIM}},
  author={Xia, Qiu and Zhang, Jun and Xu, Kaizhe and Ma, Shaodan and Jin, Shi and Yuen, Chau},
  journal={IEEE Wireless Commun. Lett.},
  year={2025},
  publisher={IEEE}
}

@inproceedings{rezvani2025uplink,
  title={Uplink wave-domain combiner for stacked intelligent metasurfaces accounting for hardware limitations},
  author={Rezvani, Maryam and Adve, Raviraj and bin Sediq, Akram and El-Keyi, Amr},
  booktitle={Proc. IEEE ICC},
  pages={},
  address={Montreal, Canada},
  year={2025},
  organization={}
}

@article{ginige2025max,
  title={Max-Min Fairness for Stacked Intelligent Metasurface-Assisted Multi-User {{MISO}} Systems},
  author={Ginige, Nipuni and Dharmawansa, Prathapasinghe and de Sena, Arthur Sousa and Mahmood, Nurul Huda and Rajatheva, Nandana and Latva-aho, Matti},
  journal={arXiv preprint arXiv:2504.14584},
  year={2025}
}

@article{fang2025stacked,
  title={Stacked Intelligent Metasurface Assisted Multiuser Communications: From a Rate Fairness Perspective},
  author={Fang, Junjie and Zhang, Chao and An, Jiancheng and Yu, Hongwen and Wu, Qingqing and Debbah, M{\'e}rouane and Yuen, Chau},
  journal={IEEE Trans. Commun.},
  volume={74},
  pages={1253--1268},
  year={2025},
  publisher={IEEE}
}

@article{dinkelbach1967nonlinear,
  title={On nonlinear fractional programming},
  author={Dinkelbach, Werner},
  journal={Management science},
  volume={13},
  number={7},
  pages={492--498},
  year={1967},
  publisher={INFORMS}
}

@article{shen2018fractional,
  title={Fractional programming for communication systems—Part I: Power control and beamforming},
  author={Shen, Kaiming and Yu, Wei},
  journal={IEEE Transactions on Signal Processing},
  volume={66},
  number={10},
  pages={2616--2630},
  year={2018},
  publisher={IEEE}
}

@article{shi2011iteratively,
  title={An iteratively weighted MMSE approach to distributed sum-utility maximization for a {MIMO} interfering broadcast channel},
  author={Shi, Qingjiang and Razaviyayn, Meisam and Luo, Zhi-Quan and He, Chen},
  journal={IEEE Transactions on Signal Processing},
  volume={59},
  number={9},
  pages={4331--4340},
  year={2011},
  publisher={IEEE}
}

@book{absil2009optimization,
  title={Optimization algorithms on matrix manifolds},
  author={Absil, P-A and Mahony, Robert and Sepulchre, Rodolphe},
  year={2008},
  publisher={Princeton University Press}
}

@article{karthiga2026multistart,
  title={Multistart Projected Gradient Descent Optimization with Armijo Backtracking-based Pre-chirp tuning for PAPR Reduction in AFDM},
  author={Karthiga, M and Deepa, D},
  journal={Physical Communication},
  pages={103035},
  year={2026},
  publisher={Elsevier}
}

@article{zhong2024ris,
  title={{RIS}-aided beamforming design for {MIMO} systems via unified manifold optimization},
  author={Zhong, Kai and Hu, Jinfeng and Li, Huiyong and Wang, Ren and An, Dongxu and Zhu, Gangyong and Teh, Kah Chan and Pan, Cunhua and Eldar, Yonina C},
  journal={IEEE Transactions on Vehicular Technology},
  volume={74},
  number={1},
  pages={674--685},
  year={2024},
  publisher={IEEE}
}

@article{paulsen1989schur,
  title={Schur products and matrix completions},
  author={Paulsen, Vern I and Power, Stephen C and Smith, Roger R},
  journal={Journal of functional analysis},
  volume={85},
  number={1},
  pages={151--178},
  year={1989},
  publisher={Elsevier}
}

\end{document}